\documentclass[letterpaper, 11 pt, conference, onecolumn]{ieeeconf} 

\newif\ifThisIsArxivVersion  
\ThisIsArxivVersiontrue

\IEEEoverridecommandlockouts             
\usepackage{amsmath,amssymb,graphicx,chemarr,xcolor}

\usepackage{subcaption}
        
\usepackage{color, xcolor}
\usepackage{tikz, graphicx}
    \usetikzlibrary{positioning}
    \usetikzlibrary{arrows}
    \usetikzlibrary{decorations.pathmorphing}
    \usetikzlibrary{arrows.meta}
    \newcommand{\cf}[1]{\textsf{#1}}
    \tikzset{%
        fwdrxn/.style={very thick, arrows={-Stealth[length=5pt,width=5pt]}},
        revrxn/.style={very thick, arrows={-Stealth[length=5pt,width=5pt,left]}},
        newt/.style={turq, opacity=0.15}
    }
    \tikzset{%
        EdgeShift/.style n args={2}{transform canvas={xshift={#1}, yshift={#2}}}, %
        Jnode/.style={circle, draw=black, thick,  fill=JnodeColor, fill opacity = 1, inner sep=0pt,minimum size=0.48cm, outer sep=2pt, font=\footnotesize},
        posEdge/.style={very thick, arrows={-Stealth[length=7.5pt,width=7.5pt]} },
        negEdge/.style={very thick, -| }
    }
    \definecolor{JnodeColor}{RGB}{235,230,237}
    \definecolor{ActivEdgeGreen}{RGB}{34,167,132}
    \definecolor{InhibEdgeRed}{RGB}{153,8,3}
    \definecolor{JnodeColorExternal}{RGB}{168, 165, 162}
    \definecolor{vblue}{RGB}{34, 69, 151}

\ifThisIsArxivVersion
\else 
\fi 
\allowdisplaybreaks

\usepackage{comment}

\newcommand{\R}{\mathbb{R}}

\newcommand{\rr}{\R}
\newcommand{\rrp}{\R_{\geq}}
\newcommand{\rrpp}{\R_{>}}
\newcommand{\beqn}{\begin{eqnarray*}}
\newcommand{\eeqn}{\end{eqnarray*}}
\newcommand{\halmos}{\rule{1ex}{1.4ex}}
\newenvironment{myproof}{\noindent {\em Proof}.\ }{\hspace*{\fill}$\halmos$\medskip}
\newcommand{\epr}{\end{myproof}}
\newcommand{\bpr}{\begin{myproof}}

\usepackage{mathtools}
\usepackage{amsthm}
\usepackage[nameinlink]{cleveref}
    \crefname{ex}{Example}{Examples}
    \crefname{thm}{Theorem}{Theorems} 
    \crefname{lem}{Lemma}{Lemmas}
    \crefname{prop}{Proposition}{Propositions}
    \crefname{cor}{Corollary}{Corollaries} 
    \crefname{conj}{Conjecture}{Conjectures} 
    \crefname{defn}{Definition}{Definitions}
    \crefname{rmk}{Remark}{Remarks} 
    \crefname{figure}{Fig.}{Figs.} 
    
	\newtheorem{thm}{Theorem}[section]
	\newtheorem{lem}[thm]{Lemma}
	
	\newtheorem{cor}[thm]{Corollary}
	
	\newtheorem*{thm*}{Theorem}
	\newtheorem*{cor*}{Corollary}
	\theoremstyle{definition} 
		\newtheorem{defn}[thm]{Definition}
		\newtheorem{ex}[thm]{Example}
    	\newtheorem{rmk}[thm]{Remark}	
\newcommand{\df}[1]{{\bf\emph{{#1}}}} 
\newcommand{\eq}[1]
    {\begin{align*}#1\end{align*}}
	\newcommand{\eqn}[1]
        {\begin{align}#1\end{align}}

\DeclareMathAlphabet\mathbfcal{OMS}{cmsy}{b}{n}
\newcommand{\mbc}[1]{\ensuremath{\mathbfcal{#1}}}
\newcommand{\vv}[1]{{\boldsymbol{#1}}}  
    \newcommand{\xx}{{\vv x}}
\newcommand{\mm}[1]{\mathbf{#1}}        
\newcommand{\kk}{k}

\newcommand\mbf[1]{\mathbf{#1}}  
\newcommand\mrm[1]{\mathrm{#1}}
\newcommand{\ehat}{\hat{\mbf e}}
\DeclareMathOperator{\sgn}{sgn}
\newcommand{\ratecnst}[1]{{\color{black}\footnotesize #1}}

\newcommand{\X}{\mathbf{X}} 

\usepackage{nicematrix}

\usepackage[colorinlistoftodos,prependcaption,
    textsize=tiny,  
    linecolor=orange, bordercolor=orange, backgroundcolor=orange!25!yellow, 
    textcolor=black]{todonotes}

\ifThisIsArxivVersion
    \newcommand{\todoMINE}[1]{\todo{#1}}
\else 
    \newcommand{\todoMINE}[1]{}
\fi

\ifThisIsArxivVersion
        \newcommand{\citeVersion}[1]{{{#1}}}

        \newcommand{\thmIFFL}{{\color{black}{IFFL}}}
        \newcommand{\thmPFBL}{{\color{black}{PFBL}}}
        \newcommand{\thmNFBL}{{\color{black}{NFBL}}}
        \newcommand{\thmIOpath}{{\color{black}{I/O path}}}

\else
        \newcommand{\citeVersion}[1]{{\color{black}{\cite[{#1}]{arxiv_self}}}}

        \newcommand{\thmIFFL}{{\color{black}{IFFL}}}
        \newcommand{\thmPFBL}{{\color{black}{PFBL}}}
        \newcommand{\thmNFBL}{{\color{black}{NFBL}}}
        \newcommand{\thmIOpath}{{\color{black}{I/O path}}}
\fi 
\newcommand{\pfthmspacing}{\ifThisIsArxivVersion\else\vspace{-0.18cm}\fi%
        }

\usepackage[resetlabels]{multibib}
    \newcites{SM}{Appendix References}

\title{\LARGE \bf
    \ifThisIsArxivVersion 
        A necessary condition for nonmonotonic dose response, 
        with an application to a kinetic proofreading model 
        -- Extended version
    \else 
        A necessary condition for nonmonotonic dose response, 
        \\
        with an application to a kinetic proofreading model 
    \fi
}

\author{Polly Y. Yu$^{1}$
  and
Eduardo Sontag$^{2}$
\thanks{This work was partially supported by grants
AFOSR FA9550-22-1-031, FA9550-21-1-0289, and NSF/DMS-2052455.}%
\thanks{$^{1}$\mbox{University of Illinois Urbana-Champaign}
{\tt\footnotesize  pollyyu@illinois.edu}}%
\thanks{$^{2}$Northeastern University
{\tt\footnotesize e.sontag@northeastern.edu}}}

\begin{document}

\maketitle
\ifThisIsArxivVersion
    \thispagestyle{plain}
    \pagestyle{plain}
\fi

\begin{abstract}
    Steady state nonmonotonic (\lq\lq biphasic\rq\rq) dose responses are often observed in experimental biology, which raises the control-theoretic question of identifying which possible mechanisms might underlie such behaviors. It is well known that the presence of an incoherent feedforward loop (IFFL) in a network may give rise to a nonmonotonic response. It has been conjectured that this condition is also necessary, i.e. that a nonmonotonic response implies the existence of an IFFL. In this paper, we show that this conjecture is false, and in the process prove a weaker version: that either an IFFL must exist or both a positive feedback loop and a negative feedback loop must exist. Towards this aim, we give necessary and sufficient conditions for when minors of a symbolic matrix have mixed signs. Finally, we study in full generality when a model of immune T-cell activation could exhibit a steady state nonmonotonic dose response.
\end{abstract}

\section{Introduction and background}
\label{sec:intro}

A dose response curve plots the steady state value of an output (the \lq\lq response\rq\rq) for a given input (the \lq\lq dose\rq\rq). A nonmonotonic, or  biphasic, dose response curve is either bell-shaped or U-shaped, characterized by low (respectively high) responses at both low dosage and high dosage. Such curves are commonly observed in biology, including the activation of immune T-cells~\cite{LeverLimKrugerNguyenEtAl2016}. In this work, we study what mechanisms underlie steady state biphasic responses. 

It is well-known that the presence of an incoherent feedforward loop (IFFL) can result in a steady state response that is nonmonotonic; see e.g., \cite{KimKwonCho2008} or the references in \cite{Sontag2010}. Negative feedback loops (NFBLs) or IFFLs are necessary, as otherwise the theory of monotone systems implies that the steady state response (\lq\lq input-to-state characteristic\rq\rq) will be monotonic on input values~\cite{AngeliSontag2003, AngeliSontag2013}. (In fact, for monotone systems, even transient responses at any given time also behave monotonically on input magnitude.)

It has been argued that IFFLs are also necessary. For example, in \cite[main text \& SI  section \lq\lq Negative feedback cannot produce a bell-shaped dose-response\rq\rq]{LeverLimKrugerNguyenEtAl2016}, the authors stated that \lq\lq models without an incoherent feed-forward loop but with negative feedback\ldots\ cannot produce a bell-shaped dose-response.\rq\rq\  That is, the authors conjectured that a biphasic response implies the existence of an IFFL. They justified this claim by performing a numerical search over network architectures, concluding in the statement that \lq\lq examining these 274 networks showed that the basic mechanism underlying all compatible networks was KPL-IFF\rq\rq\ (which is an IFFL from input to output). We show that this conjecture is false for general ODE systems by providing a counterexample. Furthermore we prove that a weaker version is true: either an IFFL must exist, or both a NFBL and a positive feedback loop (PFBL) must exist. 

Towards this aim, we define the notions of \emph{quasi-adaptation} (when the steady state response curve has a vanishing derivative), \emph{biphasic response} (when the derivative changes sign), and finally \emph{stable biphasic response} (where the relevant steady state curve is a stable branch). Quasi-adaptation and biphasic response are algebraic properties, in contrast to stable biphasic response, which is dynamical in nature. 
We also give necessary and sufficient conditions for when minors of a symbolic matrix have mixed signs, thereby giving necessary conditions for quasi-adaptation. Necessary conditions for stable biphasic response follow from monotone systems theory. The following implication diagrams summarize our results, where all the relevant terms will be defined rigorously. 
\begin{center}
\vspace{-0.07cm}
\ifThisIsArxivVersion
\begin{tikzpicture}
\else
\begin{tikzpicture}[yscale=0.77]
\fi
    \node at (-1.75,0)  {Quasi-adaptive};
    \node at (1.75,0.03) {IFFL$^1$ or PFBL$^2$}; 
    \node at (0,0) {$\implies$};
    \node at (-1.75,-0.44)  {$\Uparrow$};
    \node at (-1.75,-1)  {Biphasic}; 
    \node at (-1.75,-2) {Stable biphasic};
    \node at (1.75,-2+0.04) {IFFL$^1$ or NFBL$^3$};
    \node at (-1.75,-1.44)  {$\Uparrow$};
    \node at (0,-2) {$\implies$};
\end{tikzpicture}
\vspace{-0.18cm}
\end{center}
\footnotetext[1]{An IFFL from the input $x_1$ to the output $x_j$, defined only when $j\neq 1$.}%
\footnotetext[2]{When $j=1$, a PFBL that is disjoint from the node $x_1$. When $j \neq 1$, a PFBL that is vertex-disjoint from an input-output path.}%
\footnotetext[3]{A NFBL that is reachable from the input $x_1$ and to the output $x_j$.}%
\stepcounter{footnote}%
\stepcounter{footnote}%
\stepcounter{footnote}%

Finally, inspired by \cite{LeverLimKrugerNguyenEtAl2016}, we consider a model for T-cell activation that consists of two components: a kinetic proofreading network whereby antigens bind to T-cell receptors (TCRs), and a downstream network consisting solely of activation and inhibition (defined in \Cref{sec:activation-inhibition}). We show that in order for an output from the downstream network to exhibit stable biphasic response to antigen level, this network necessarily either contains an IFFL, or it contains both a PFBL and a NFBL.

\subsection{Notations}

The following notations are used throughout this paper.
\begin{itemize}
    \item $\rrp^n$ and $\rrpp^n$ are sets of vectors with nonnegative and positive components respectively. 
    \item $\ehat_i$ is the $i$th orthonormal vector of Euclidean space. 
    \item $\ehat_{ij}$ is the matrix with a $1$ in the $(i,j)$ position, and 0 everywhere else. 
    \item  $\mm J(\gamma, \delta)$ is the submatrix of $\mm J$ with rows indexed by $\gamma$ and columns by $\delta$; $\hat i$ refers to all indices except $i$. If $\gamma = \delta$, we write $\mm J(\gamma)$. 
    \item $\mm J[\gamma, \delta] = \det \mm J(\gamma, \delta)$. We write $\mm J[\gamma]$ for $\mm J[\gamma, \gamma]$. 
    \item $[n] = \{ 1,2,\ldots, n\}$.
\end{itemize}

In \Cref{sec:mixed-signs}, we work exclusively with a \df{signed symbolic matrix} $\mbc J$, whose $(i,j)$ entry is either $a_{ij}$, $-a_{ij}$, or $0$, where $a_{ij}$ is a variable. A polynomial, e.g., $\det \mbc J$, is said to have \df{mixed signs} if it has both a positive and a negative term.

\section{Signed symbolic matrix}
\label{sec:mixed-signs}

Here, we give necessary and sufficient conditions for when a (principal or non-principal) minor of a signed symbolic matrix with negative diagonals
\begin{align*}
    \mbc J = \begin{pNiceArray}{ccccc}
        - & * & \cdots & \cdots & * \\ 
        * & - & * & \cdots & * \\ 
        \vdots &  & \Block{1-2}<>{\raisebox{5pt}{\rotatebox[origin=c]{10}{$\ddots$}}} && \vdots \\
        * & \cdots & \cdots & * & -
    \end{pNiceArray}
\end{align*}    
has mixed signs. Typically, $\mbc J$ comes from the Jacobian matrix $\mm J(\vv x)$ of an ODE system $\dot{\vv x} = \vv f(\vv x)$; in such cases, we assume $[\mm J(\vv x)]_{ij} = \partial_j f_i(\vv x)$ has constant sign for all $\vv x$.

\subsection{Graph of a signed symbolic matrix}

Let $\mbc J$ be a $n \times n$ signed symbolic  matrix with negative diagonals. It is associated to its \df{J-graph} $G$, a digraph with $n$ nodes and signed edges in the typical sense: for any nodes $i, j$, there is a positive (resp. negative) edge from $i$ to $j$ if $[\mbc J]_{ji} > 0$ (resp. $[\mbc J]_{ji} < 0$); otherwise there is no edge from $i$ to $j$. We denote such an edge as $(i,j)$. By assumption on $\mbc J$, every node has a negative self-loop. 

We recall some commonly used terms. A \df{walk} is a nonempty sequence of directed edges joining a sequence of vertices. A \df{path} is a simple walk, i.e., no repeating vertices possibly with the exception of the first and last node, in which case, it is a \df{cycle}. For a subset of edges $E$, its \df{sign} is $\sgn(E) \coloneqq \prod_{e\in E} \sgn(e)$. 
A subset of nodes $X$ is said to be \df{reachable from node $i$} (resp. \df{to node $j$}) if for each $x \in X$, there exists a walk from $i$ to $x$ (resp. from  $x$ to $j$).
A subset of edges $E$ is a \df{disjoint cycle cover} of $G$ if each connected component is a cycle and $E$ covers all the nodes of $G$. Note that cycles in a disjoint cycle cover are vertex-disjoint, and some cycles may be self-loops. 

A \df{feedback loop}  is a cycle $C$ of length at least two, i.e., $|C| \geq 2$. It is \df{positive} (resp. \df{negative}) if $\sgn(C) > 0$ (resp. $\sgn(C) < 0$). A \df{feedforward loop from $i$ to $j$} is a pair of paths $P_1 \neq P_2$ that originate from node $i$, and terminate at node $j \neq i$. It is  \df{coherent} (resp. \df{incoherent})  if  the signs of $P_1,P_2$ are the same (resp. opposite). 
For simplicity, we refer to the above as PFBL, NFBL, CFFL, and IFFL. With $i$  as input, and $j \neq i$ as output, a path from $i$ to $j$ is an \df{input-output path} (I/O path). 

Finally, by a subgraph of $G$, we mean a subset of edges and all vertices incident on them.  Denote by $G(\hat i,\hat j)$ the subgraph obtained by deleting all incoming edges to $i$ and all outgoing edges from $j$. Where $\gamma$ is a subset of nodes, by $G(\gamma)$ we mean the subgraph obtained by keeping only the nodes in $\gamma$ and edges incident on those nodes. For example $G(\hat 1)$ is the subgraph obtained by deleting the node $1$ along with any edges coming into or going out of $1$.

\subsection{Mixed signs in minors of a signed symbolic matrix}

Consider a principal minor $\mbc J[\gamma]$.  Previous works on multiple steady states gave conditions for when $\mbc J[\gamma]$ has mixed signs~\cite{Banaji2010, Soule2003}. 

\begin{lem} 
\label{lem:symbolic-principal_minor}
    Let $\mbc J$ be a $n \times n$ signed symbolic  matrix with negative diagonals, and $G$ be its J-graph. For any $\gamma \subseteq [ n]$, the polynomial $\mbc J[\gamma]$ has mixed signs if and only if $G(\gamma)$ has a \thmPFBL.
\end{lem}
\begin{proof}
    Since the proof for $\gamma \subsetneq [n]$ is similar, we assume $\gamma = [n]$.  By Leibniz formula, $\det \mbc J = \sum_{\sigma \in S_n} T_\sigma$, where the monomial $T_\sigma \coloneqq \sgn(\sigma) \prod_{i=1}^n [\mbc J]_{\sigma(i), i}$ is nonzero if and only if $[\mbc J]_{\sigma(i),i} \neq 0$ for all $i$. Thus any nonzero term $T_\sigma$ is in a one-to-one correspondence with a disjoint cycle cover $E_\sigma$. The monomials for different $\sigma$'s are algebraically independent, so $\det \mbc J$ has mixed signs if and only if $\sgn(T_\sigma)\sgn(T_\eta) = -1$ for some $\sigma, \eta \in S_n$. 
    
    For any $T_\sigma\neq 0$, let $\sigma = \tau_1\cdots \tau_p$ be its nontrivial cycle decomposition, and $\Theta :=\{ i : \sigma(i) = i\}$ its set of fixed points. Recall that $\sgn(\sigma) = \prod_j \sgn(\tau_j)$ and $\sgn(\tau_j) = (-1)^{|\tau_j|+1}$. As $[\mbc J]_{\sigma(i),i}\neq 0$ if and only if $(i,\sigma(i))$ is an edge in $G$, each $\tau_j$ corresponds to a FBL $C_j$ in $E_\sigma$, and each $j \in \Theta$ corresponds to a self-loop. 
    The sign of $T_\sigma$ is given by 
    \ifThisIsArxivVersion
        \begin{align*}
            \sgn(T_\sigma) = 
            \left[ \prod_{i \in \Theta} \sgn([\mbc J]_{i,i}) \right] 
                \left[\prod_{j=1}^p (-1)^{|\tau_j|+1} \sgn(C_j) \right]
            = (-1)^{n+p} \prod_{j=1}^p \sgn(C_j), 
        \end{align*}
    \else
        $\left[ \prod_{i \in \Theta} \sgn([\mbc J]_{i,i}) \right] 
            \left[\prod_{j=1}^p (-1)^{|\tau_j|+1} \sgn(C_j) \right] 
        = (-1)^{n+p} \prod_{j=1}^p \sgn(C_j)$, 
    \fi
    which depends on the number of NFBL among the $p$ feedback loops in $E_\sigma$. Thus $\sgn(T_\sigma) = (-1)^{n} (-1)^{\#\text{PFBL in $E_\sigma$}}$. 
    
    The fact that every node has a self-loop implies that there is a term $T_\mrm{id}$ with sign $(-1)^n$. Clearly the lack of PFBL implies that all terms have sign $(-1)^n$. Conversely, if there is at least one PFBL $C$, then $C$ and self-loops on nodes not traversed by $C$ together form a disjoint cycle cover, whose corresponding term has sign $(-1)^{n+1}$. 
\end{proof}

\begin{lem}
\label{lem:symbolic-non-principal_minor}
    Let $\mbc J$ be a $n \times n$ signed symbolic  matrix with negative diagonals, and $G$ be its J-graph. Designate $i$ as input and $j$ as output with $j \neq i$.  
    \begin{enumerate}
        \item The non-principal minor $\mbc J[\hat i, \hat j]$ is identically zero if and only if $G$ has no \thmIOpath. 
        \item Suppose $G$ has an \thmIOpath. Then $\mbc J[\hat i, \hat j]$ has mixed signs if and only if either $G$  has an \thmIFFL\ from $i$ to $j$, or $G$ has a \thmPFBL\  and an \thmIOpath\ that are vertex-disjoint.
    \end{enumerate}
\end{lem}
\begin{proof}
    Without loss of generality, let $i = 1$ and $j = n$, so
    \begin{align*}
        \mbc J[\hat 1, \hat n]  &= (-1)^{n-1}\det
        \begin{pNiceArray}{ccc|c}
            0 & \cdots & 0 & {1} \\ 
            \hline 
            \Block{3-3}<>{\mbc J(\hat 1, \hat n)} & & & 0 \\ 
            &&& \vdots \\ 
            &&& 0 
        \end{pNiceArray}.
    \end{align*}
    Denote the above matrix by $\mbc A$. Its J-graph $H$ can be obtained from $G(\hat 1, \hat n)$ adding a positive edge from $n$ to $1$. By Leibniz formula, $\mbc J[\hat 1, \hat n] = (-1)^{n+1} \sum_{\sigma \in S_n} \sgn(\sigma) [\mbc A]_{\sigma(i),i}$, where we denote each term as $T_\sigma$. Clearly $T_\sigma \neq 0$ if and only if $\sigma(n) = 1$, in which case let $\sigma = \tau_1 \cdots \tau_p$ be its nontrivial cycle decomposition where $\tau_1(n) = 1$. Since any I/O path $P$ in $G$ is bijectively associated to the cycle  $P \cup (n,1)$ in $H$, we conclude that $\mbc J[\hat 1,\hat n] \not\equiv 0$ if and only if an I/O path is present in $G$.
    
    Now suppose $G$ has an I/O path, and consider any $T_\sigma \neq 0$, which corresponds to a disjoint cycle cover $\tilde{E}_\sigma$ of $H$. Then $E_\sigma \coloneqq \tilde{E}_\sigma \setminus (n,1)$
    is a vertex-disjoint collection of cycles and an I/O path in $G(\hat 1, \hat n)$. Each $T_\sigma \neq 0$ is uniquely associated to such a $E_\sigma$. Clearly there is a one-to-one correspondence between self-loops in $\tilde E_\sigma$ and those in $E_\sigma$; between the cycles $C_2,\ldots, C_p$ in $G(\hat 1, \hat n)$ and $H$; between $C_1$ and $P$. 
    Because 
    \ifThisIsArxivVersion
        \begin{align*}
        \sgn(T_\sigma) = (-1)^{n-1} \left[ (-1)^{n+p} \sgn(P) \prod_{j=2}^p \sgn(C_j)  \right] ,
        \end{align*}
    \else 
        $\sgn(T_\sigma) = (-1)^{n-1} \left[ (-1)^{n+p} \sgn(P) \prod_{j=2}^p \sgn(C_j)  \right] $, 
    \fi
    we have $\sgn(T_\sigma) = \sgn(P) (-1)^{\#\text{PFBL in $E_\sigma$}}$. 

    To prove one direction of 2, suppose two I/O paths $P, Q$ form an IFFL. Choosing $P$ and self-loops on all nodes not visited by $P$ gives a vertex-disjoint collection $E_\sigma$ that covers all nodes. Similarly, let $E_\eta$ be such a collection containing $Q$. As neither collection has any FBL at all, $\sgn(T_\sigma) \sgn(T_\eta) = \sgn(P) \sgn(Q) = -1$. If instead, suppose there is a PFBL $C$ that is vertex-disjoint from the I/O path $P$. The subgraph $E_\xi$ consisting of $P, C$ and self-loops for all remaining nodes is associated to a nonzero term $T_\xi$. Moreover, $\sgn(T_\sigma) \sgn(T_\xi) = -1$. 

    For the other direction, suppose $\mbc J[\hat 1, \hat n]$ has mixed signs, say $T_\sigma T_\eta < 0$. Let $E_\sigma, E_\eta$ be the corresponding subgraphs, where $E_\sigma$ contains the I/O path $P$, and $E_\eta$ contains $Q$. Either $P$ and $Q$ form an IFFL, or else $-1 = (-1)^{\#\text{PFBL in $E_\sigma$} \,+\,\, \#\text{PFBL in $E_\eta$}}$, 
    so there is a PFBL $C$ in either $E_\sigma$ or $E_\eta$, disjoint from $P$ or $Q$ respectively. 
\end{proof}

\section{Steady state response of an input-output system}
\label{sec:biphasic}

In this section, we consider the steady state response curves of input-output systems where the control $u$ affects only one variable; without loss of generality, let it be $x_1$.

Throughout this work, we consider the following input-output system with output $y = x_j$ for some $j$:  
\begin{align}
\label{eq:syst_main}
    \dot{\vv x} &= {\vv F(\vv x, u) \coloneqq{} } 
        \vv f(\vv x) + \ehat_1 {g(u)},  
\end{align}
where $\vv f : \X \subset \R^n \to \R^n$ is $C^1$, $u \in J$ for some $J \subset \rr$. We assume that for all $\vv x\in \X$, the entries in the Jacobian matrix $\frac{\partial \vv f}{\partial \vv x}$ have constant signs, and $\frac{\partial f_i}{\partial x_i} < 0$ for all $i$. We also assume $\partial_u g \neq 0$ is constant in sign. 
We denote a solution to \eqref{eq:syst_main} with initial state $\xx(0)$ by $\phi(t,\xx(0),u)$. 

In addition, we assume
\ifThisIsArxivVersion
    \begin{align}
    \label{eq:assumption-A0}
        \forall u_0 \in J \,\, \exists \vv x_0\in \X \text{ such that } {\vv F}(\vv x_0, u_0) = \vv 0,  
        \text{ and } \frac{\partial \vv F}{\partial \vv x}(\vv x_0, u_0) \text{ has full rank}.
     \tag{A0}
    \end{align}
    Sometimes we make the stronger assumption:
    \begin{align}
    \begin{gathered}
    \label{eq:assumption-A1}
        \forall u_0 \in J \,\, \exists \vv x_0\in \X \text{ such that } {\vv F}(\vv x_0, u_0) = \vv 0, 
        \,
        \frac{\partial \vv F}{\partial \vv x}(\vv x_0, u_0) \text{ is Hurwitz, and } \\ 
        \exists  \epsilon>0 \colon \forall u \in B_\epsilon(u_0)  \,\,  \forall \vv x(0)\in \X\,\,
        \exists \vv x^*(u)\in \X \colon  
         {\vv F}(\vv x^*(u), u) = \vv 0 \text{ and } \phi(t,\vv x(0), u) \to \vv x^*(u).
    \end{gathered} \tag{A1}
    \end{align}
\else 
    \begin{align}
    \begin{gathered}
    \label{eq:assumption-A0}
        \forall u_0 \in J \,\, \exists \vv x_0\in \X \text{ such that } {\vv F}(\vv x_0, u_0) = \vv 0,  
        \\ 
        \text{and } \frac{\partial \vv F}{\partial \vv x}(\vv x_0,u_0) \text{ has full rank}.
    \end{gathered} \tag{A0}
    \end{align}
    Sometimes we make the stronger assumption:
    \begin{align}
    \begin{gathered}
    \label{eq:assumption-A1}
        \forall u_0 \in J \,\, \exists \vv x_0\in \X \text{ such that } {\vv F}(\vv x_0, u_0) = \vv 0, 
        \\ 
        \frac{\partial \vv F}{\partial \vv x}(\vv x_0, u_0) \text{ is Hurwitz, and } \\ 
        \exists  \epsilon>0 \colon \forall u \in B_\epsilon(u_0)  \,\, \forall \vv x(0)\in \X\,\, 
        \exists \vv x^*(u)\in \X \colon  
        \\
         {\vv F}(\vv x^*(u), u) = \vv 0 \text{ and } \phi(t,\vv x(0), u) \to \vv x^*(u). 
    \end{gathered}\tag{A1}
    \end{align}
\fi%
Assumption (A0) insures the existence of smooth curves of equilibria: by the Implicit Function Theorem (IFT), we know that for every $u_0$ and every $\xx_0$ such that ${\vv F}(\xx_0,u_0) = \vv 0$, there is an $\epsilon>0$ and a continuously differentiable steady state curve $\vv x^*(u)$ for $u \in B_\epsilon(u_0)$ with $\vv x^*(u_0)=\xx_0$.

The set of assumptions in (A0) and (A1) are somewhat redundant. Clearly (A1) implies (A0), since a Hurwitz matrix has full rank. In addition, since eigenvalues depend continuously on matrix entries, the Hurwitz property implies that  $\vv x^*(u)$, insured to exist by the IFT, has the property that $\vv x^*(u)$ is asymptotically stable for each $u$ in a possibly smaller neighborhood of $B_{\epsilon}(u_0)$. Moreover, one could use a converse Lyapunov theorem to guarantee an invariant region for small variations in the nominal $u_0$.

As is often done, we can depict the input and output of \eqref{eq:syst_main} on the J-graph of $\frac{\partial \vv f}{\partial \vv x}$, by adding an input edge and an output edge. This graph is the \df{J-graph of the input-output system} \eqref{eq:syst_main}.

Intuitively, a nonmonotonic steady state response for $x_j$ requires $\partial_u x_j^*$ to change sign. We define \emph{quasi-adaptation} as when the derivative vanishes, a phenomenon also known as \emph{infinitesimal homeostasis}~\cite{GolubitskyStewart2017}.
\begin{defn}
\label{def:quasi-adaptative}\label{def:biphasic}
    The input-output system \eqref{eq:syst_main} under assumption (A0) with output $x_j$ is said to be \df{quasi-adaptive} if there exists $u_0 \in J$ such that $\partial_u x_j^*(u_0) = 0$. The system is said to exhibit \df{biphasic response} if there exist $u_1, u_2 \in J$ such that $[\partial_u x_j^*(u_1)][\partial_u x_j^*(u_2)] < 0$.  
\end{defn} 

Quasi-adaptation and biphasic response are properties of a branch of steady states. There is no a priori assumption of uniqueness or stability. \Cref{ex:biphasic_unstable} gives a system with two branches of steady states, where only the unstable branch is biphasic (\Cref{fig:biphasic_unstable-ss_curves}). For the purpose of applications, one is often interested in a stable biphasic branch. 
\begin{defn}
\label{def:biphasic-stable}
    The system \eqref{eq:syst_main} under assumption (A1) with output $x_j$ is said to exhibit \df{stable biphasic response} if $x_j^*(u)$ is biphasic and asymptotically stable.
\end{defn}
\pfthmspacing

\subsection{Conditions for quasi-adaptation and biphasic response}

We give necessary conditions for quasi-adaptation and  biphasic response under assumption (A0). We treat separately the cases when the input and output are the same vs. distinct.

\begin{lem}
\label{lem:Cramer}
    Consider the system  \eqref{eq:syst_main} under assumption (A0) with output $x_j$. Let $\vv x^*(u)$ be a steady state curve defined for all $u \in J$, and $\mm J^* \coloneqq \frac{\partial \vv f}{\partial \vv x}(\vv x^*(u))$. Then 
    \begin{align}\label{eq:lem-Cramer}
        \frac{\partial x_j^*}{\partial u} =  
            (-1)^j \partial_u g \,
            \frac{
            \mm J^*[\hat 1, \hat j]}{\det \mm J^*} . 
    \end{align}
\end{lem}
\begin{proof}
    Along the steady state curve, $\vv 0  = \vv f(\vv x^*) + \ehat_1 g(u)$. Implicitly differentiating, $\vv 0 = \mm J^* \partial_u \vv x^* + \ehat_1 \partial_u g$, which when solved using Cramer's rule gives the result.
\end{proof}

Let  $\mbc J$ be the signed symbolic matrix consistent with $\mm J(\vv x) = \frac{\partial \vv f}{\partial \vv x}$, i.e., $[\mbc J]_{ij} > 0$ if and only if $\frac{\partial f_i}{\partial x_j} > 0$ and similarly for zero and negative entries. 
If $\mbc J[\hat 1, \hat j] \equiv 0$, then $\mm J(\vv x)[\hat 1, \hat j] = 0$ for all $\vv x \in \X$. For $\mm J(\xx)$ coming from \eqref{eq:syst_main}, this implies that $x_j^*(u)$ is independent of $u$, i.e., $x_j^*$ exhibits \lq\lq perfect adaptation\rq\rq\ (in control-theoretic terms, zero DC gain, or disturbance rejection of constant disturbances).  
Quasi-adaptation, thus biphasic response, implies $\mbc J[\hat 1, \hat j]$, if nontrivial, has mixed signs.

\begin{lem}
\label{lem:alg}
    Consider the system  \eqref{eq:syst_main} under assumption (A0) with output $x_j$. Let $\mbc J$ be the signed symbolic matrix consistent with $\frac{\partial \vv f}{\partial \vv x}$. If the polynomial $\mbc J[\hat 1, \hat j]$ is not identically zero and it has no mixed signs, then the system cannot be quasi-adaptive, nor can it exhibit biphasic response. 
\end{lem}
\begin{proof}
    Under assumption (A0),  $\partial_u x^*_j$ is nonsingular, so the denominator in the expression of $\partial_u x^*_j$ in \eqref{eq:lem-Cramer} cannot be zero or change sign. Thus, quasi-adaptation occurs if and only if the numerator vanishes. Similarly, biphasic response requires the numerator to change sign, which by continuity of $\partial_u x^*_j$, necessarily means vanishing. 
    Thus, quasi-adaptation and biphasic response both depend on the minor $\mm J^*[\hat 1, \hat j]$ of the Jacobian matrix $\frac{\partial \vv f}{\partial \vv x}$. Because of the assumption on $\mbc J[\hat 1,\hat j]$, $\mm J^*[\hat 1, \hat j]$ is either always positive or always negative.
\end{proof}

We proved more than what we claimed in \Cref{lem:alg}:  if $\mbc J[\hat 1, \hat j]$ is not identically zero and has no mixed signs, then as a function of $\vv x$, $\frac{\partial \vv f}{\partial \xx}[\hat 1, \hat j]$ can never vanish.

\begin{thm}
\label{thm:graph-QA-i=j}
    Consider the system  \eqref{eq:syst_main} under assumption (A0) with output $x_1$, and let $G$ be its J-graph. 
    If $G$ has no \thmPFBL\ disjoint from the node $x_1$, the system cannot be quasi-adaptive, nor can it  exhibit biphasic response. 
\end{thm}
\begin{proof}
    Let $\mbc J$ be the signed symbolic matrix consistent with  $\frac{\partial \vv f}{\partial \vv x}$.  By \Cref{lem:alg}, it suffices to show that $\mbc J[\hat 1, \hat 1] \not\equiv 0$ and it has no mixed signs. 
    Let $\tilde{G}$ be the J-graph of $\frac{\partial \vv f}{\partial \vv x}$, and thus the J-graph of $\mbc J$. 
    By \Cref{lem:symbolic-principal_minor}, $\mbc J[\hat 1, \hat 1]$ has no mixed if and only if $\tilde G(\hat 1)$ has no PFBL. This is equivalent to $G$ having no PFBL that is disjoint from the input/output node $x_1$, since $\tilde G$ is $G$ without the input/output edges.
\end{proof}

\pfthmspacing
\begin{thm}
\label{thm:graph-QA-i!=j}
    Consider the system \eqref{eq:syst_main} under assumption (A0) with output $x_j$ with $j\neq 1$, and let $G$ be its J-graph. 
    \begin{enumerate}
    \item If $G$ has no \thmIOpath, then the steady state response $x_j^*(u)$ is independent of  $u$. 
    \item Suppose $G$ has an \thmIOpath. If $G$ neither has an \thmIFFL\ from input to output, nor a \thmPFBL\ that is vertex-disjoint from some \thmIOpath, then the system cannot be quasi-adaptive, nor can it exhibit biphasic response. 
    \end{enumerate}
\end{thm}
\begin{proof}
    Without loss of generality, let $x_n$ be the output. Let $\mbc J$ be the signed symbolic matrix consistent with $\mm J^* \coloneqq \frac{\partial \vv f}{\partial \xx}(\vv x^*(u))$, and let $\tilde G$ be the J-graph of $\mbc J$, obtained from $G$ by deleting the input and output edges. If $G$ has no I/O path, neither does $\tilde G$, so $\mbc J[\hat 1, \hat n] \equiv 0$ by \Cref{lem:symbolic-non-principal_minor}, hence $\mm J^*[\hat 1, \hat n] = 0$ for all $u$. Claim 1 follows from \Cref{lem:Cramer}. 

    Suppose $G$ has an I/O path, so $\mbc J[\hat 1, \hat n] \not\equiv 0$. By \Cref{lem:alg}, it suffices to show that $\mbc J[\hat 1, \hat n]$ has no mixed signs. The subgraphs listed in  2 are in $G$ if and only if they are in $\tilde G$, so by \Cref{lem:symbolic-non-principal_minor}, $\mbc J[\hat 1, \hat n]$ has no mixed signs. 
\end{proof}

\pfthmspacing
\begin{rmk}
\label{rmk:remarks-classical_examples}
    The classical adapting integral feedback and IFFL linear systems (see for example \cite[Section 6.1]{ShovalAlonSontag2011}) do not fall under our results. For the integral feedback system, whose J-graph is shown in \Cref{fig:remarks-classical_examples-integralFBL}, there is no self-loop on $x$. The IFFL linear system, whose J-graph is in \Cref{fig:remarks-classical_examples-IFFL}, is a system where $u$ affects more than one variable. 
\end{rmk}

\begin{figure}[t!]
\centering 
\begin{subfigure}[t]{0.22\textwidth}
    \centering
    \vspace{1.5cm}
    \begin{tikzpicture}[xscale=1, yscale=0.75, transform canvas={xshift=-0.6cm, yshift=0.5cm}]
    \node at (-0.75,0) {}; 
            \node[Jnode] (1) at (0,0) {$x$};  
            \node[Jnode] (2) at (1.25,0) {$y$}; 
            \draw[negEdge, EdgeShift={0pt}{2.75pt}] (1)--(2);
            \draw[posEdge, EdgeShift={0pt}{-2.75pt}] (2)--(1);
            \draw[negEdge] (2.north east) to[out=50, in=0, looseness=6] (2.east);

            \draw[posEdge] (1.25,1) -- (2); 
            \draw[posEdge] (2)--(1.25,-1);
    \end{tikzpicture}
    \caption{}
    \label{fig:remarks-classical_examples-integralFBL}
\end{subfigure}
\begin{subfigure}[t]{0.22\textwidth}
    \centering
    \vspace{1.5cm}
    \begin{tikzpicture}[xscale=1, yscale=0.75, transform canvas={xshift=-0.6cm,yshift=0.5cm}]
            \node[Jnode] (1) at (0,0) {$x$};  
            \node[Jnode] (2) at (1.25,0) {$y$}; 
            \draw[negEdge] (1)--(2);
            \draw[negEdge] (2.north east) to[out=50, in=0, looseness=6] (2.east);
            \draw[negEdge] (1.west) to[out=180, in=230, looseness=6] (1.south west);

            \draw[posEdge] (0,1) -- (1); 
            \draw[posEdge] (1.25,1)--(2);
            \draw[posEdge] (2)--(1.25,-1);
    \end{tikzpicture}
    \caption{}
    \label{fig:remarks-classical_examples-IFFL}
\end{subfigure}
\caption{
    Our theorems are silent on the (a) adapting integral feedback and (b) IFFL linear systems.
}
\label{fig:remarks-classical_examples}
\vspace{-0.5cm}
\end{figure}
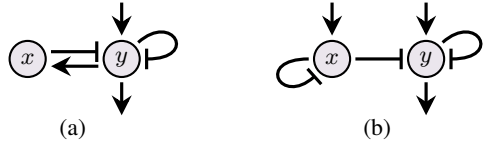
\pfthmspacing

We present counterexamples to the claim in \cite{LeverLimKrugerNguyenEtAl2016} that an IFFL is necessary for stable biphasic response.

\begin{ex}
\label{ex:biorxiv2}
    This example first appeared in \cite{Sontag2020_biphasic_biorxiv}. Consider the  system on $\rrp^4$ with control $u \in [0.4,0.6]$ and output $x_3$ (and a different system with output $x_4$):%
    \begin{align}
    \begin{split}
    \label{eq:ex:biorxiv2}
        \dot x_1 &= -x_1 + f(x_2) \\ 
        \dot x_2 &= h(x_1) - x_2 + x_3 \\ 
        \dot x_3 &= h(x_2) -2x_3 + u\\ 
        \dot x_4 &= x_3 - x_4 , 
    \end{split}
\intertext{where $f(s) = e^{1-\sigma(s)}$, $h(s) = \frac{1}{2}e^{2(1-\sigma(s))}$, and   } 
        \sigma(s) &= \left\{\begin{array}{cl}%
        {0.8}&{0\leq s<0.8} \\ %
        {s  }&{0.8\leq s\leq 1.2} \\%
        {1.2}&{1.2\leq s}
        \end{array}\right. . \nonumber 
    \end{align}%
Both $x_3$ and $x_4$ exhibit stable biphasic response curves with $x_3^*(u) = x_4^*(u)$, as shown in \Cref{fig:biorxiv2-ss_curve}. Moreover, its J-graph (\Cref{fig:biorxiv2-Jgraph}) has no IFFL. Whether the output is $x_3$ or $x_4$, the PFBL between $x_1$ and $x_2$ and the NFBL between $x_2$ and $x_3$ are necessary for stable biphasic response. (Full detail can be found in \citeVersion{\Cref{sec:app-ex:bioarxiv2}}.)
\end{ex}

\begin{figure}[ht]
\vspace{-0.3cm}
\centering 
\begin{subfigure}[t]{0.23\textwidth}
    \centering
    \begin{tikzpicture}[xscale=1, yscale=0.8, transform canvas={xshift=-0.6cm, yshift=2.55cm}]
    \node at (-1.25,-3.15) {}; 
            \node[Jnode] (1) at (1.25,-2.5) {$x_1$};  
            \node[Jnode] (2) at (0,-2.5) {$x_2$}; 
            \node[Jnode] (3) at (0,-1.25) {$x_3$}; 
            \node[Jnode] (4) at (1.25,-1.25) {$x_4$}; 
            \draw[negEdge] (4.north east) to[out=50, in=0, looseness=6] (4.east);
            \draw[negEdge] (1.north east) to[out=50, in=0, looseness=6] (1.east);
            \draw[negEdge] (3.north west) to[out=130, in=180, looseness=6] (3.west);
            \draw[negEdge] (2.north west) to[out=130, in=180, looseness=6] (2.west);
            \draw[posEdge, EdgeShift={-2.75pt}{0pt}] (0,-0.25) -- (3);  
            \draw[posEdge] (3) -- (4); 
            \draw[negEdge, EdgeShift={0pt}{-2.75pt}] (1)--(2);
            \draw[negEdge, EdgeShift={0pt}{2.75pt}] (2)--(1);
            \draw[posEdge, EdgeShift={-2.75pt}{0pt}] (3)--(2);
            \draw[negEdge, EdgeShift={2.75pt}{0pt}] (2)--(3);
            
            \draw[posEdge, dashed, vblue!30!blue] (4) -- (1.25,-0.25); 
            \draw[posEdge, EdgeShift={2.75pt}{0pt}] (3) -- (0,-0.25); 
    \end{tikzpicture}
    \caption{}
    \label{fig:biorxiv2-Jgraph}
\end{subfigure}
\begin{subfigure}[t]{0.23\textwidth}
    \centering
    \includegraphics[height=1.in]{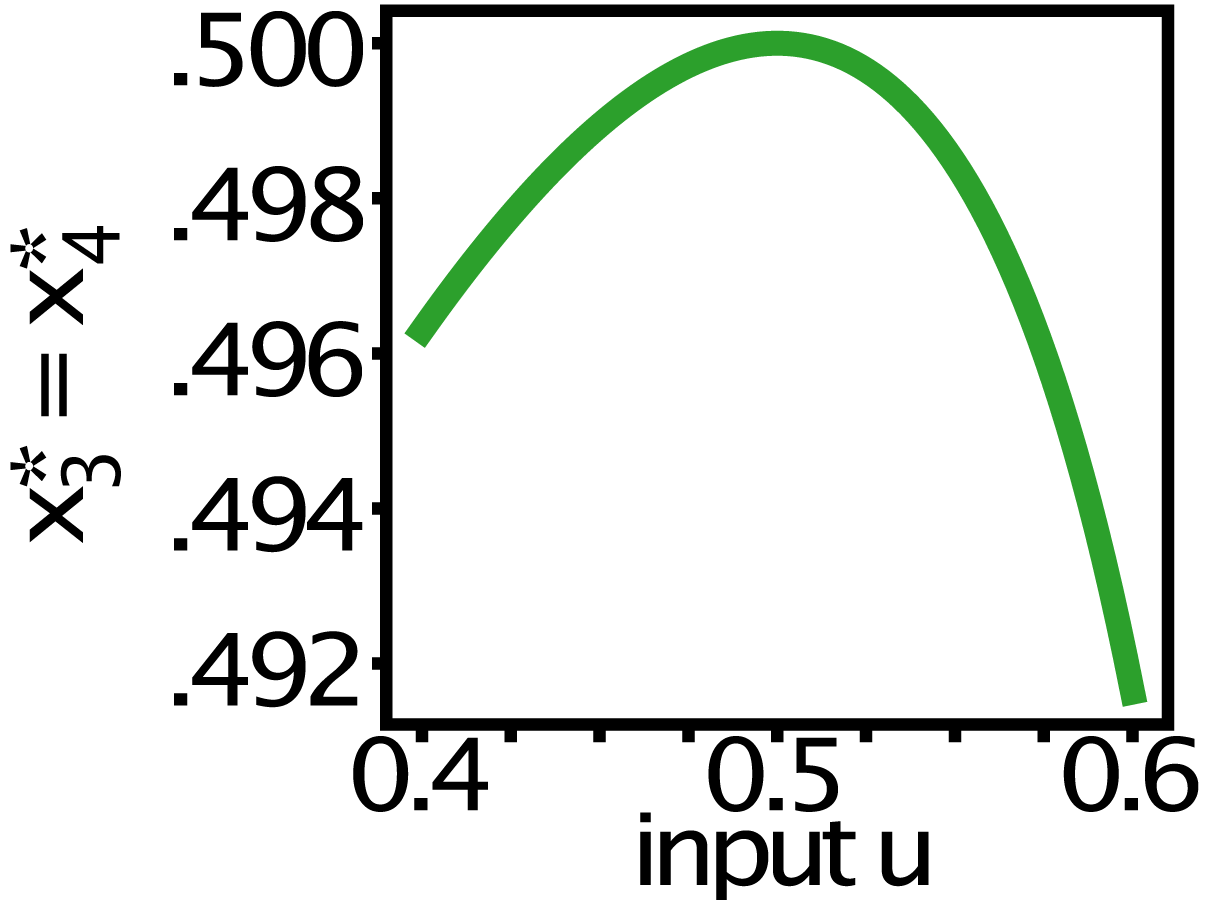} 
    \caption{}
    \label{fig:biorxiv2-ss_curve}
\end{subfigure}
\caption{
    (a) The J-graph of the system from \Cref{ex:biorxiv2}, with $x_3$ being the output of one, and $x_4$ the other (blue, dashed). 
    (b) The system exhibits stable biphasic response in $x_3$ and $x_4$. The \thmPFBL\ between $x_1$ and $x_2$ and \thmNFBL\ between  $x_2$ and $x_3$ are necessary for the nonmonotonic response. 
}
\label{fig:biorxiv2}
\end{figure}
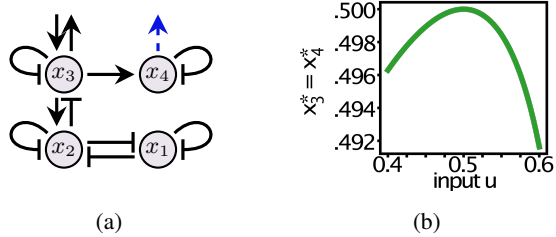
\pfthmspacing

\begin{ex}
\label{ex:biphasic_unstable}
    This example illustrates that \Cref{thm:graph-QA-i=j,thm:graph-QA-i!=j} do not assume stability. Consider the following system on $\rrp^4$: 
    \begin{align}
    \begin{split}
    \label{eq:ex:biphasic_unstable}
        \dot x_1 &= u - 2 x_1      \\
        \dot x_2 &=   x_1 + x_3 -2 x_2^2  + {x_3^2}/{4}   \\
        \dot x_3 &=   
            -{7x_3}/{4}  + x_2^2    \\
        \dot x_4 &=  x_1    + {3x_3}/{4}   - x_4, 
    \end{split}
    \end{align}
    with control $u \in [0,12.5]$ and output  $x_4$. Its J-graph (\Cref{fig:biphasic_unstable-Jgraph}) does not have an IFFL. The system has two branches of steady states (\Cref{fig:biphasic_unstable-ss_curves}), one stable and the other unstable, which is not monotonic. Without the PFBL between $x_2$ and $x_3$, this system cannot exhibit biphasic response. (Full detail can be found in \citeVersion{\Cref{sec:app-ex:biphasic_unstable}}.)
\end{ex}

\begin{figure}[ht]
\vspace{-0.3cm}
\centering 
\begin{subfigure}[t]{0.18\textwidth}
    \centering
    \begin{tikzpicture}[xscale=1, yscale=0.8, transform canvas={xshift=-0.63cm, yshift=1.55cm}]
    \node at (0,-2) {};
            \node[Jnode] (1) at (0,0) {$x_1$}; 
            \node[Jnode] (4) at (1.25,0) {$x_4$}; 
            \node[Jnode] (2) at (0,-1.25) {$x_2$}; 
            \node[Jnode] (3) at (1.25,-1.25) {$x_3$}; 
            \draw[negEdge] (4.north east) to[out=50, in=0, looseness=6] (4.east);
            \draw[negEdge] (3.north east) to[out=50, in=0, looseness=6] (3.east);
            \draw[negEdge] (1.north west) to[out=130, in=180, looseness=6] (1.west);
            \draw[negEdge] (2.north west) to[out=130, in=180, looseness=6] (2.west);
            \draw[posEdge] (0,1) -- (1);  
            \draw[posEdge] (4) -- (1.25,1);  
            \draw[posEdge] (1) -- (4);
            \draw[posEdge](1)--(2); 
            \draw[posEdge] (3)--(4);
            \draw[posEdge, EdgeShift={0pt}{2.75pt}] (2)--(3);
            \draw[posEdge, EdgeShift={0pt}{-2.75pt}] (3)--(2);
    \end{tikzpicture}
    \caption{}
    \label{fig:biphasic_unstable-Jgraph}
\end{subfigure}
\begin{subfigure}[t]{0.23\textwidth}
    \centering 
    \includegraphics[height=1.in]{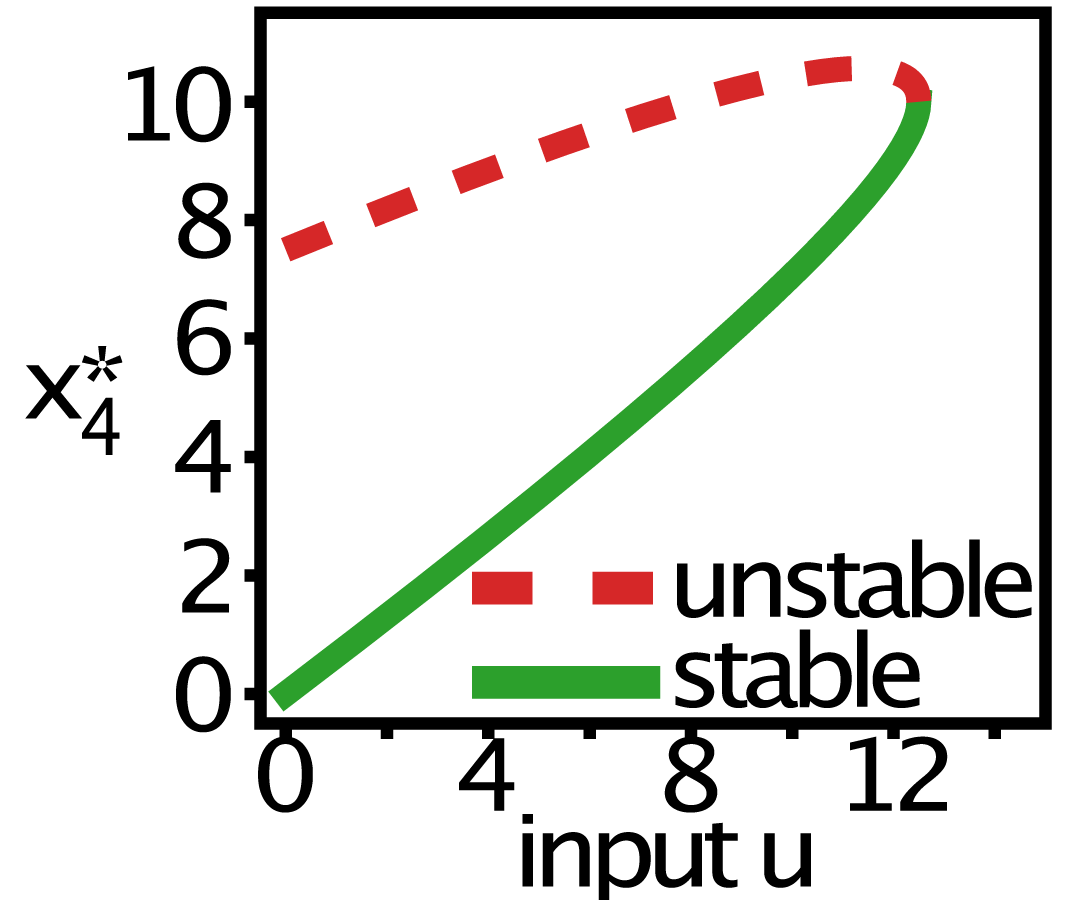} 
    \caption{}
    \label{fig:biphasic_unstable-ss_curves}
\end{subfigure}
\caption{
    (a) The J-graph of the system from \Cref{ex:biphasic_unstable}. 
    (b) The system has two branches of steady states: a stable monotontic and an unstable nonmonotonic branch. The \thmPFBL\ between $x_2$ and $x_3$ is necessary for biphasic response. 
}
\label{fig:biphasic_unstable}
\end{figure}
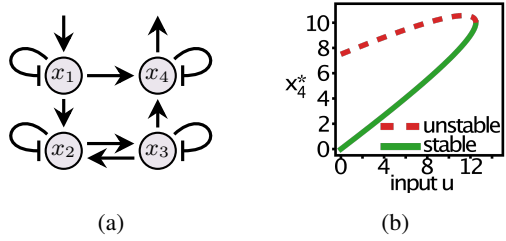
\pfthmspacing

\subsection{Conditions for stable biphasic response}

The theory of monotone systems predicts that stable branches of steady states will be nondecreasing or nonincreasing depending on the net sign of the \lq\lq input-output characteristic\rq\rq\ of the system, while unstable branches may be decreasing or increasing, irrespective of the sign of the characteristic. These facts are well illustrated by the example of a genetic autoregulatory transcription network analyzed in~\cite{enciso_sontag_2008}. In Fig.~3a of that paper, one can see four stable branches (drawn in solid blue lines) as well as four unstable branches of equilibria (drawn in dotted red lines); of the four unstable branches, three are decreasing functions of $u$.

To obtain necessary conditions for \emph{stable} biphasic responses, we use Theorem~1 from \cite{AscensaoDattaHanciogluSontagEtAl2016}, which concerns a system $\dot{\vv x} = \vv F(\vv x, u(t))$, initially at steady state, with piecewise continuous input $u(t)$ that is monotonic in time. It states that if all walks (with no self-loops) from the input node $u$ to the output node $x_j$ have the same sign\footnote{%
    In \cite{AscensaoDattaHanciogluSontagEtAl2016}, the control $u$ is given its own node in the graph. In our context where $u$ only affects $x_1$, this is equivalent to adding a node for $u$ and an edge $(u, x_1)$ with sign  $\sgn(\partial_u g)$. Thus $x_j(t)$ is monotonic if all walks (with no self-loops) from $x_1$ to $x_j$ have the same sign.
    }%
, then $x_j(t)$ is monotonic in time. Moreover, suppose $u(t)$ is nondecreasing in time and $\partial_u g(u(t)) > 0$, then  $x_j(t)$ is nondecreasing (resp. nonincreasing) if any I/O path $P$ is positive (resp. negative). 
More precisely, at any $t \geq 0$, the sign of $\partial_t x_j(t)$ is either $0$ or given by $\sgn(\partial_t u) \sgn(\partial_u g) \sgn(P)$.

\begin{thm}
\label{thm:stable-biphasic}
    Consider the  system  \eqref{eq:syst_main} under assumption (A1), and let $G$ be its J-graph. Let $\xx^*(u)$ be a stable branch of steady states. 
    \begin{enumerate}
        \item Suppose $x_1$ is the output. If there is no \thmNFBL\ in $G$ that is reachable to and from the node $x_1$, then $x_1^*(u)$ is monotonic. 
    
        \item Suppose $x_j$ is the output with $j \neq 1$. If $G$ neither has an \thmIFFL\ from $x_1$ to $x_j$, nor a \thmNFBL\ that is reachable from $x_1$ and to $x_j$, then $x_j^*(u)$ is monotonic. 
    \end{enumerate}
    Under the assumption stated above, the output cannot exhibit stable biphasic response. 
\end{thm}
\begin{proof}
    We first show that the conditions listed in the theorem imply that any walk (with no self-loops) from $x_1$ to $x_j$ has the same sign.  Then by \cite[Theorem 1]{AscensaoDattaHanciogluSontagEtAl2016}, the solution $x_j(t)$ is monotonic in time.  For the rest of this proof, by a walk, we mean a walk without traversing any self-loops.

    (1)  There are two oppositely signed walks from $x_1$ to $x_1$ if and only if there is a NFBL $C$ in the strongly connected component of $x_1$. If $P_1$ is a path from $x_1$ to node $x_i$ in $C$ and $P_2$ is a path from $x_i$ to $x_2$, then $P_1 \cup P_2$ and $P_1 \cup C \cup P_2$ are two oppositely signed walks from $x_1$ to $x_1$.

    (2) If $G$ does not have an I/O path, then $x_j^*$ cannot be biphasic by \Cref{thm:graph-QA-i!=j}, so suppose otherwise. Clearly an IFFL from $x_1$ to $x_j$ constitutes two oppositely signed walks. A walk from $x_1$ to $x_j$ that traverses a NFBL generates a different walk than with the NFBL omitted; these are two walks with opposite signs from $x_1$ to $x_j$.
    
    Therefore, under the hypotheses listed in 1 or 2, $x_j(t)$ is nondecreasing (resp. nonincreasing) if $\partial_u g > 0$ and $u(t)$ is non-decreasing\footnote{If $\partial_t u \leq 0$ (or alternatively $\partial_u g < 0$), then the sign of $\dot x_j(t)$ is flipped.}, 
    and any walk from $x_1$ to $x_j$ is positive (resp. negative) by \cite[Theorem 1]{AscensaoDattaHanciogluSontagEtAl2016}. In what follows, we assume that any walk from $u$ to $x_j$ is positive; the other cases only require appropriate changes in signs and inequalities.

    Finally, we show that $x_j(t)$ being monotonic in $t$ implies  $x_j^*(u)$ is monotonic in $u$. 
    The system is $\dot{\vv x} = \vv f(\vv x) + \ehat_1 {g(u(t))}$, where $u(t) = u_0 + \epsilon$ for a fixed $\epsilon > 0$ such that $u_0 + \epsilon \in J$. 
    Suppose the system is initially at $\xx^*(u_0)$; let $\vv x(t)$ be a solution, defined for all $t \geq 0$. Since $x_j(t)$ is monotonically nondecreasing,  $x_j^*(u_0+\epsilon) = \lim_{t\to \infty} x_j(t) \geq x_j^*(u_0)$, where the convergence to $x_j^*(u_0+\epsilon)$ is guaranteed by assumption (A1). So the $C^1$ steady state curve must have nondecreasing derivative at $u_0$. Since $u_0$ and $\epsilon$ are arbitrary, we conclude that $\partial_u x_{j}^*(u) \geq 0$ for all $u \in J$.  
\end{proof}

The implication diagram in the introduction is a corollary of \Cref{thm:graph-QA-i=j,thm:graph-QA-i!=j,thm:stable-biphasic}.

\begin{cor}
\label{cor:summary} 
    Consider the  system \eqref{eq:syst_main} under assumption (A1). Necessary conditions for stable biphasic response are:
    \begin{itemize}
        \item when $x_1$ is output, the existence of a \thmPFBL\ disjoint from the node $x_1$ and a \thmNFBL\ that is reachable to and from the node $x_1$; 
        \item when $x_j$ is output with $j \neq 1$, either (a) the existence of an \thmIFFL\ from $x_1$ to $x_j$, or (b) the existence of a \thmPFBL\ vertex-disjoint from some \thmIOpath\ and a \thmNFBL\ that is reachable from $x_1$ and to $x_j$. 
    \end{itemize}
\end{cor}

\section{Alternative forms of control}
\label{sec:activation-inhibition}

The examples presented thus far involve the simplest form of control: $g(u) = u$. 
In the context of biochemistry, such a term with $u \geq 0$ might represent inflow or production of the species modelled by $x_1$. However, other forms of controls occur naturally. For example, suppose $x_1$ represents the concentration of a protein $\cf{X}_1$ in its \emph{active} form, and its inactive form could be enzymatically activated by another species, e.g., see \cite{LeverLimKrugerNguyenEtAl2016,RameshSuwanmajoKrishnan2023}. Then the control would take the form $g(u,x_1) =  u(T_1 - x_1)$, where $u$ is proportional to the concentration of this other species, and $T_1$ is the total amount of $\cf{X}_1$ in the system. Throughout this work, we assume the activation/inhibition reactions follow mass-action kinetics.

\begin{defn}
\label{def:activation-inhibition}
    We say $\cf{C}$ \df{activates} (resp. \df{inhibits}) $\cf{X}$ to refer to the reactions $\cf{Z} \rightleftharpoons \cf{X}$ and $\cf{C} + \cf{Z} \to \cf{C} + \cf{X}$ (resp. $\cf{C} + \cf{X} \to \cf{C} + \cf{Z}$), along with the conservation law $z + x = T$, where $\cf{Z}$ is the inactive form of $\cf{X}$.
\end{defn}

We now consider a biochemical system involving species $\cf{X}_1,\ldots, \cf{X}_n$. We assume each $\cf{X}_i$ has an active form (with concentration $x_i$) and an inactive form (with concentration $z_i$), and that $\cf{X}_1$ is activated (or inhibited) by an external species $\cf{C}$. The other reactions allowed are $\cf{X}_i$ being activated or inhibited by other $\cf{X}_j$'s in their active forms. For any $i \in [n]$, let
\eq{ 
    \Lambda_\text{act}^i &\coloneqq \left\{ j \in [n] \colon \cf{X}_j \text{ activates } \cf{X}_i ,\,  j \neq i \right\} , 
    \\ 
    \Lambda_\text{inh}^i &\coloneqq \left\{ j \in [n] \colon \cf{X}_j \text{ inhibits } \cf{X}_i ,\,  j \neq i \right\} , 
}
where $\Lambda_\text{act}^i \cap \Lambda_\text{inh}^i = \emptyset$. We are interested in the case when $\bigcup_{i=1}^n \left( \Lambda_\text{act}^i \cup \Lambda_\text{inh}^i \right) \neq \emptyset$. 
We called such a system with only activation and/or inhibition reactions an \df{activation-inhibition network}. 
See \citeVersion{\Cref{sec:app-internal-activation-inhibition}} for a concrete example.

Let $k_c > 0$ be a catalytic rate constant; if $\cf{X}_1$ is activated by $\cf{C}$, let $\tilde{g}(C,x_1, z_1) = k_c C z_1$, or in the case of inhibition, $\tilde{g}(C,x_1,z_1) = - k_c C x_1$.  
Under mass-action kinetics, the dynamics of an activation-inhibition network where $\cf{X}_1$ is activated (or inhibited) by an external species $\cf{C}$ is given by 
\ifThisIsArxivVersion
    \eq{ 
        -\dot z_1 = \dot x_1 
        &=  k_\text{on}^1 z_1 - k_\text{off}^1 x_1 
            + \! \sum_{j \in \Lambda_\text{act}^1} \! k_\text{cat}^{1j} x_j z_1 
            - \! \sum_{j \in \Lambda_\text{inh}^1} \! k_\text{cat}^{1j} x_j x_1  + \tilde{g}(C, x_1, z_1)
        \\ 
        -\dot z_i = \dot x_i 
        &=  k_\text{on}^i z_i - k_\text{off}^i x_i 
            + \! \sum_{j \in \Lambda_\text{act}^i} \! k_\text{cat}^{ij} x_j z_i 
            - \! \sum_{j \in \Lambda_\text{inh}^i} \! k_\text{cat}^{ij} x_j x_i  ,
            \qquad \text{ for } i \neq 1, 
    }
    where 
\else 
    $-\dot z_i = \dot x_i 
        =  k_\text{on}^i z_i - k_\text{off}^i x_i 
            +  \sum_{j \in \Lambda_\text{act}^i} \! k_\text{cat}^{ij} x_j z_i 
            -  \sum_{j \in \Lambda_\text{inh}^i} \! k_\text{cat}^{ij} x_j x_i  
        + \delta_{i1} \tilde{g}(C, x_1, z_1)
    $, where $\delta_{ij}$ is the Kronecker delta and %
\fi 
$k_\text{on}^i, k_\text{off}^i$, and $k_\text{cat}^{ij} > 0$  are constants. 
For each $i \in [n]$, there is a conservation law $x_i + z_i = T_i$ for some constant $T_i > 0$. 
Hence, we can eliminate the variables $z_i$. The resulting system  $\dot{\vv x} = \vv f(\vv x) + \ehat_1 g(C,x_1)$ is given by  
\ifThisIsArxivVersion
\eqn{\begin{split}\label{eq2:activation-inhibition-general} 
    \dot x_1 &=  \left( k_\text{on}^1 + \!\sum_{j \in \Lambda_\text{act}^1} \! k_\text{cat}^{1j} x_j \right) (T_1 - x_1)  
        - \left( k_\text{off}^1 + \! \sum_{j \in \Lambda_\text{inh}^1} \! k_\text{cat}^{1j} x_j \right)x_1
        +  g(C, x_1)
         \\ 
    \dot x_i &= \left( k_\text{on}^i + \! \sum_{j \in \Lambda_\text{act}^i} \! k_\text{cat}^{ij} x_j\right) (T_i - x_i) 
        - \left( k_\text{off}^i + \! \sum_{j \in \Lambda_\text{inh}^i} \! k_\text{cat}^{ij} x_j \right)x_i,
        \qquad \text{ for } i \neq 1, 
\end{split}}
\else 
\eqn{\begin{split}\label{eq2:activation-inhibition-general} 
    \dot x_i &= \left( k_\text{on}^i + \! \sum_{j \in \Lambda_\text{act}^i} \! k_\text{cat}^{ij} x_j\right) (T_i - x_i) 
        \\&\quad 
        - \left( k_\text{off}^i + \! \sum_{j \in \Lambda_\text{inh}^i} \! k_\text{cat}^{ij} x_j \right)x_i 
    + \delta_{i1}  g(C, x_1)  , 
\end{split}}
\fi%
where $g(C, x_1) = k_c C(T_1 - x_1)$ in the case of activation, and $g(C,x_1) = -k_c Cx_1$ in the case of inhibition. 

We make two observations about \eqref{eq2:activation-inhibition-general}. First, $0 \leq x_i \leq T_i$ for all $i$, because $\rrp^{2n}$ is forward-invariant for $(\vv x, \vv z)$, and $x_i + z_i = T_i$. Second, the system has no boundary steady state (i.e., either $x_i = 0$ or $x_i = T_i$). This follows because if $x_i = 0$, then $\dot x_i > 0$; if $x_i = T_i$, then $\dot x_i < 0$. For mass-action systems, the strictly positive orthant is forward-invariant~\cite[Lemma II.1]{Sontag2001}; therefore, we assume $0 < z_i, x_i < T_i$ for all $i$. 

We can compute the Jacobian matrix $\mm J$ of $\vv f(\vv x)$ directly: 
\eq{ 
    \frac{\partial f_i}{\partial x_i} 
        &=  - k_\text{on}^i - k_\text{off}^i -  \!\!\!\!\! \sum_{j \in \Lambda_\text{act}^i \cup \Lambda_\text{inh}^i} \!\!\! k^{ij}_\text{cat} x_j, 
}\eq{
    \frac{\partial f_i}{\partial x_j} 
        &= \left\{ 
            \begin{array}{cl}
                k_\text{cat}^{ij} (T_i - x_i)   & \text{ if } j \in \Lambda_\text{act}^i, \\ 
                -k_\text{cat}^{ij} x_i  & \text{ if } j \in \Lambda_\text{inh}^i, \\ 
                0 & \text{ otherwise},
            \end{array}
        \right. 
        \qquad \text{ for } j \neq i.
}
The diagonal entries of $\mm J$ are negative. The off-diagonal $(i,j)$ entry is positive if $j \in \Lambda_\text{act}^i$ since $T_i - x_i > 0$, and negative if $j \in \Lambda_\text{inh}^i$; $[\mm J]_{ij} \equiv 0$ if $j \not\in \Lambda_\text{act}^i \cup \Lambda_\text{inh}^i$. In other words, $\mm J$ is a signed matrix with negative diagonals. Note that when $\cf{X}_1$ is activated by $\cf{C}$, $\frac{\partial g}{\partial C} = k_c(T_1 - x_1) > 0$, and $\frac{\partial g}{\partial C} = -k_c x_1 < 0$ in the case of inhibition. 

As a result of the signed structure of $\mm J$, its J-graph $G$ is well-defined. Moreover, we can define the J-graph of the  system \eqref{eq2:activation-inhibition-general} by adding an output edge from $x_j$ and an input edge to $x_1$, where the input edge is positive (resp. negative) if $\cf{X}_1$ is activated (resp. inhibited) by $\cf{C}$. The resulting graph is the familiar one in biology for activation/inhibition.

The system \eqref{eq2:activation-inhibition-general} has the same form as \eqref{eq:syst_main}, but for a specialized control function $g$. Assuming \eqref{eq2:activation-inhibition-general} satisfies assumption \eqref{eq:assumption-A0} and that $C$ is at steady state, we can carry out the same analysis as in \Cref{lem:Cramer}, and obtain a formula for $\partial_C x_j^*$. Since $\vv 0 = \vv f(\vv x^*(C)) + \ehat_1 g(C, x_1^*(C))$,  $\vv 0 = \left[\frac{\partial \vv f}{\partial \vv x}(\vv x^*(C)) \right]\partial_C \vv x^* - \ehat_1 k_c\left(x_1^* + C \partial_C x_1^* - \delta_\text{act}T_1\right) $, where $\delta_\text{act} = 1$ for activation, and $0$ for inhibition. 
Therefore, in the case of activation, 
\eq{ 
  \frac{\partial x_j^*}{\partial C} &= 
     (-1)^j k_c (T_1 - x_1^*) \frac{ \mm J^*[\hat 1, \hat j]}{\det \left( \mm J^* - \ehat_{11} k_c C\right) }    , 
\intertext{and in the case of inhibition,} 
  \frac{\partial x_j^*}{\partial C} &= 
    (-1)^{j+1} k_c  x_1^* \frac{ \mm J^*[\hat 1, \hat j]}{\det \left( \mm J^* - \ehat_{11} k_c C\right)} ,  
}
where $\mm J^* \coloneqq \frac{\partial \vv f}{\partial \vv x}(\vv x^*(C))$. Comparing these formulas with \eqref{eq:lem-Cramer} indicates that the input edge's sign as defined corroborate with the interpretation of activation/inhibition.

Since $x_1^* \neq 0, T_1$ for any $C$, whether $\partial_C x_j^*$ vanishes or changes sign depends only on $\mm J^*[\hat 1, \hat j]$. Our results in the previous section give necessary conditions for (stable) biphasic response of $x_j$ to the steady state concentration of the external species $\cf{C}$. 

\begin{cor}
\label{cor2:activation-inhibition}
    The necessary conditions for stable biphasic response stated in \Cref{cor:summary} apply to the system \eqref{eq2:activation-inhibition-general} under assumption \eqref{eq:assumption-A1}. 
\end{cor}

\section{Application to immune T-cell activation}
\label{sec:application-KP_ss}

In this section, we study in full generality a model of immune T-cell activation. Following \cite{LeverLimKrugerNguyenEtAl2016}, our model  consists of two parts. The first is a network for kinetic proofreading (\Cref{fig:KP-core}), whereby an antigen $\cf{L}$ binds to a T-cell receptor (TCR) $\cf{R}$, and the resulting TCR complexes $\cf{C}_i$ can be at various stages of phosphorylation. The second half of the model is a downstream input-output system, where one or more of $\cf{C}_i$'s act as input to the species $\cf{X}_1$. We give necessary conditions on the downstream network for when the output of this network could exhibit stable biphasic response.  

In what follows, we let $L(t), R(t)$ denote the concentrations of  \cf{L} and \cf{R}. Similarly, let $C_i(t)$ be the concentration of  $\cf{C}_i$, and $C_T(t) \coloneqq \sum_{j=0}^N C_i(t)$. Where there is no ambiguity, we drop the explicit dependence on $t$. The kinetic proofreading system has two conservation laws: $L_T = L + C_T$ and $R_T = R + C_T$, with $L_T, R_T > 0$.  

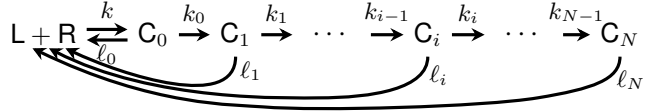
\begin{figure}[ht]
\centering 
\vspace{2.25cm}
    \begin{tikzpicture}[xscale=0.73, yscale=0.6, transform canvas={xshift=-2.35cm, yshift=1.2cm}]
        \begin{scope}
		    \node (i) at (-2,0) {\small $\cf{L}+\cf{R}$}; 
		    \node (0) at (0,0) {\small $\cf{C}_0$};
		    \node (1) at (1.5,0) {\small $\cf{C}_1$};
		    \node (2) at (3.25,0) {\phantom{\,\,}$\cdots$\phantom{\,\,}};
		    \node (3) at (5,0) {\small $\cf{C}_i$};
		    \node (4) at (6.75,0) {\phantom{\,\,}$\cdots$\phantom{\,\,}};
		    \node (5) at (8.5,0) {\small $\cf{C}_N$};
            \draw [fwdrxn, EdgeShift={0pt}{2pt}] (i)--(0) node [midway, above] {\ratecnst{$\kk$}};
            \draw [fwdrxn] (0)--(1) node [midway, above] {\ratecnst{$\kk_{0}$}};
            \draw [fwdrxn] (1)--(2) node [midway, above] {\ratecnst{$\kk_{1}$}};
            \draw [fwdrxn] (2)--(3) node [midway, above] {\ratecnst{$\kk_{i-1}$}};
            \draw [fwdrxn] (3)--(4) node [midway, above] {\ratecnst{$\kk_{i}$}};
            \draw [fwdrxn] (4)--(5) node [midway, above] {\ratecnst{$\kk_{N-1}$}};
            \draw [fwdrxn, EdgeShift={0pt}{-2pt}] (0) --(i)  node [midway, below=-2.5pt] {\ratecnst{$\ell_{0}$}} ;
            \draw [fwdrxn] (1.south) to[out=270, in=330] (-1.6,-0.3);
                \node at  (1.4,-0.8) [right] {\ratecnst{$\ell_{1}$}};
            
            \draw [fwdrxn] (3.south) .. controls (5,-1.2) and (4,-1.4) .. (3,-1.4)
            .. controls (2.5,-1.4) and (0,-1.5) .. (-1.9,-0.3); 
                \node at (5.2,-0.9) {\ratecnst{$\ell_{i}$}};
            
            \draw [fwdrxn] (5.south) .. controls (8.5,-1) and (8,-1.2) .. (7,-1.4)
            .. controls (6,-1.6) and (3,-1.7) .. (2,-1.6)
            .. controls (1.5,-1.55) and (0,-1.6) .. (-2.2,-0.3);
                \node at (8.7,-1) {\ratecnst{$\ell_{N}$}};
        \end{scope}
    \end{tikzpicture}
\caption{Kinetic proofreading network.}
\label{fig:KP-core}
\ifThisIsArxivVersion\else \vspace{-0.5cm}\fi
\end{figure}

Assume the kinetic proofreading network in \Cref{fig:KP-core} evolves according to mass-action kinetics: $\dot{\vv z} = \vv F(\vv z; L_T)$ where $\vv z = (L,R, C_0,\ldots, C_N)^\top$. The form of $\vv F(\vv z; L_T)$ is shown in \citeVersion{\eqref{eq:app-KP_core_ODE} of \Cref{sec:app-KP_core}}. 
This system has a globally attracting steady state within its stoichiometric compatibility class for any $L_T, R_T > 0$~\cite{Sontag2001}. The steady state response ${C}_i^*(L_T)$ is monotonically increasing; we derive in \citeVersion{\Cref{sec:app-application-KP_ss}} an analytical expression for the response: $C_i^* =  A_i C_T^*$, where  
\begin{align}
\label{eq:KP-CT}
\ifThisIsArxivVersion
    C_T^* = \frac{L_T + R_T + \frac{A}{\kk}}{2} - \sqrt{ \left( \frac{L_T + R_T + \frac{A}{\kk}}{2}\right)^2 - L_T R_T} , 
\else 
    C_T^* = \frac{L_T \!+\! R_T \!+\! \frac{A}{\kk}}{2} - \!\!\sqrt{\! \left( \frac{L_T \!+\! R_T \!+\! \frac{A}{\kk}}{2}\right)^2 \!\!- L_T R_T} , 
\fi
\end{align}
and $A_i, A, k > 0$ are constants. \citeVersion{\Cref{sec:app-CT_monotone}} contains a proof that $\frac{\partial C_i^*}{\partial L_T} > 0$ for all $L_T > 0$.

\subsection{Stable biphasic response in immune T-cell activation}

Inspired by \cite{LeverLimKrugerNguyenEtAl2016}, we consider a system where one or more of the TCR complexes $\cf{C}_i$ either all activate or all inhibit $\cf{X}_1$, and the species $\cf{X}_i$ could either activate or inhibit each other.  The full system has the form 
\begin{align}
\begin{split}
\label{eq:KP-full}
    \dot{\vv z} &= \vv F(\vv z; L_T)  \\
    \dot{\vv x} &= \vv f(\vv x) + \ehat_1 g(\vv z, x_1),
\end{split}
\end{align}
where $\vv f(\vv x)$ comes from an activation-inhibition network (e.g., see \eqref{eq2:activation-inhibition-general}), and $g(\vv z, x_1)$ is either $\sum_{i \in \Lambda} \kk_i C_i(T_1 - x_1)$ in the case of activation, or $- \sum_{i \in \Lambda} \kk_i C_ix_1$ in the case of inhibition, for some $\emptyset \neq \Lambda \subseteq \{0,1,\ldots, N\}$. 
We make the same assumptions as before: the entries of $\frac{\partial \vv f}{\partial \vv x}$ are constant in sign,  the diagonals are negative, and \eqref{eq:KP-full} satisfies (A1).

\begin{thm}
\label{thm:KP-full}
    Consider \eqref{eq:KP-full} with activation (respectively inhibition) of $\cf{X}_1$. Let  $(\vv z^*(L_T), \vv x^*(L_T))$ be a stable branch of steady states for $L_T$ in some interval $J \subset \rrpp$. Let $G$ be the J-graph of $\frac{\partial \vv f}{\partial \vv x}$.  
    Suppose $x_j^*(L_T)$ is stable biphasic. 
    \begin{enumerate}
        \item If $j=1$, then $G$ has a \thmPFBL\ and a \thmNFBL.
        \item If $j \neq 1$, then $G$ either has an \thmIFFL, or it has a \thmPFBL\ and a \thmNFBL. 
    \end{enumerate}
\end{thm}
\begin{proof}
    Since $\vv z^*(L_T)$ is globally stable for any $L_T > 0$, we are interested only in the downstream system. Consider an alternative system $\dot{\vv x} = \vv f(\vv x) + \ehat_1 g^*(C_T^*, x_1)$ where $g^*(C_T^*, x_1) = K C_T^*(T_1-x_1)$ in the case of activation, and $-KC_T^*x_1$ in the case of inhibition, where $K \coloneqq \sum_{i \in \Lambda} k_i A_i > 0$ and $C_T^*$ is given by \eqref{eq:KP-CT}. The original ODEs for $\vv x$ and this alternative system share $\vv x^*(L_T)$ as a stable branch of steady states. The alternative system is precisely of the form \eqref{eq2:activation-inhibition-general}, but where the activating/inhibiting \lq\lq species\rq\rq\ $C_T$ is at steady state. Applying \Cref{cor2:activation-inhibition} to this alternative system  gives us the result we desire. 
\end{proof}

\section{Discussion} 
\label{sec:discussion}

We were initially motivated to provide a rigorous proof for the claim in \cite{LeverLimKrugerNguyenEtAl2016}: that IFFL is necessary for what we now call stable biphasic response. Against this claim, we gave counterexamples (e.g., \Cref{ex:biorxiv2}), which were not among the 58'905 networks computationally explored in \cite{LeverLimKrugerNguyenEtAl2016}. Instead we proved a weaker version: that either an IFFL must exist, or both a PFBL and a NFBL must exist. 

Our results are phrased for systems with one control and one input variable, but as \Cref{thm:KP-full} shows, there are easy extensions, such as when the different controls are all scalar multiples of a function. Furthermore, if a control $u$ affects two or more variables, we can redefine $u$ as a variable, and formally introduce a new control $v$, i.e., $x_0 \coloneqq u$,  $\dot x_0 = -x_0 + v$ where $x_0(0) = v$. Then our results hold for the extended system $(x_0, \vv x)$. Clearly in this case, the J-graph should include $x_0$ as a variable node. See \citeVersion{\Cref{sec:app-multiple_inputs}} for details.

\bibliographystyle{IEEEtran}
\bibliography{2024_cdc_polly_eduardo.bib}{}

\clearpage 
\newpage 
\onecolumn

\appendix
\renewcommand\thefigure{\thesubsection\arabic{figure}}
\setcounter{figure}{0}   
\renewcommand\theequation{\thesubsection\arabic{equation}}
\setcounter{equation}{0} 
\renewcommand\thethm{\thesubsection\arabic{thm}}
\setcounter{thm}{0} 

\captionsetup{margin=0.1\textwidth}

\bigskip \bigskip

\subsection{\sc An example of activation-inhibition network} 
\label{sec:app-internal-activation-inhibition}

\begin{figure}[h]
\centering
        \begin{tikzpicture}
            \node[Jnode, fill=white] (c) at (0,0) {$\cf{C}$}; 
            \node[Jnode] (x) at (2,0) {$x_1$}; 
            \node[Jnode] (z) at (4,0) {$x_3$}; 
            \node[Jnode] (y) at (3,-1.3) {$x_2$};
            
            \draw[posEdge, ActivEdgeGreen] (x) -- (z) node [midway, above] {\scriptsize activation};  
            \draw[posEdge, ActivEdgeGreen] (z) -- (y) node[midway, below right] {\scriptsize activation};  
            \draw[posEdge, ActivEdgeGreen] (c)--(x) node[midway, above] {\scriptsize activation};
            \draw[negEdge, InhibEdgeRed] (y)--(x) node[midway, below left] {\scriptsize inhibition\!\!};
            \draw[posEdge] (z)--(5.5,0); 
        \end{tikzpicture}
\caption{An activation-inhibition network with $3$ species and external species $\cf{C}$.}
\label{fig:activation-inhibition-network}
\end{figure}
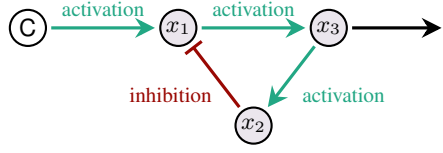

We give a concrete example of an activation-inhibition network as described in \Cref{sec:activation-inhibition}. Consider the network given by the J-graph in \Cref{fig:activation-inhibition-network}. This represents the following reactions 
    \begin{align}
    \begin{split}
    \label{eq:app-activation-inhibition-CRN} 
    \begin{tikzpicture}
    \node at (0,0.5) {};
    \begin{scope}[EdgeShift={-5.75cm}{0pt}]
        \begin{scope}[EdgeShift={0cm}{0pt}] 
            \node (a) at (0,0) [left] {$\cf{Z}_1$};
            \node (b) at (1,0) [right] {$\cf{X}_1$};
            \draw[fwdrxn, EdgeShift={0pt}{2pt}] (a)--(b) node [midway, above] {\ratecnst{$\kk_\text{on}^1$}}; 
            \draw[fwdrxn, EdgeShift={0pt}{-2pt}] (b)--(a) node [midway, below] {\ratecnst{$\kk_\text{off}^1$}};
        \end{scope} 
        \begin{scope}[EdgeShift={3cm}{0pt}] 
            \node (a) at (0,0) [left] {$\cf{Z}_2$};
            \node (b) at (1,0) [right] {$\cf{X}_2$};
            \draw[fwdrxn, EdgeShift={0pt}{2pt}] (a)--(b) node [midway, above] {\ratecnst{$\kk_\text{on}^2$}}; 
            \draw[fwdrxn, EdgeShift={0pt}{-2pt}] (b)--(a) node [midway, below] {\ratecnst{$\kk_\text{off}^2$}};
        \end{scope}
        \begin{scope}[EdgeShift={6cm}{0pt}] 
            \node (a) at (0,0) [left] {$\cf{Z}_3$};
            \node (b) at (1,0) [right] {$\cf{X}_3$};
            \draw[fwdrxn, EdgeShift={0pt}{2pt}] (a)--(b) node [midway, above] {\ratecnst{$\kk_\text{on}^3$}}; 
            \draw[fwdrxn, EdgeShift={0pt}{-2pt}] (b)--(a) node [midway, below] {\ratecnst{$\kk_\text{off}^3$}};
        \end{scope}
        \begin{scope}[EdgeShift={9.75cm}{0pt}] 
            \node (a) at (0,0) [left] {$\cf{C} + \cf{Z}_1$};
            \node (b) at (1,0) [right] {$\cf{C} + \cf{X}_1 $};
            \draw[fwdrxn] (a)--(b) node [midway, above] {\ratecnst{$\kk_{c}$}}; 
        \end{scope}
    \end{scope}
    \begin{scope}[EdgeShift={-5.5cm}{-1.25cm}]
        \begin{scope}[EdgeShift={0cm}{0pt}] 
            \node (a) at (0,0) [left] {$\cf{X}_1 + \cf{Z}_3$};
            \node (b) at (1,0) [right] {$\cf{X}_1 + \cf{X}_3$};
            \draw[fwdrxn] (a)--(b) node [midway, above] {\ratecnst{$\kk_{13}$}}; 
        \end{scope}
        \begin{scope}[EdgeShift={5cm}{0pt}] 
            \node (a) at (0,0) [left] {$\cf{X}_3 + \cf{Z}_2$};
            \node (b) at (1,0) [right] {$\cf{X}_3 + \cf{X}_2$};
            \draw[fwdrxn] (a)--(b) node [midway, above] {\ratecnst{$\kk_{32}$}}; 
        \end{scope}
        \begin{scope}[EdgeShift={10cm}{0pt}] 
            \node (a) at (0,0) [left] {$\cf{X}_2 + \cf{X}_1$};
            \node (b) at (1,0) [right] {$\cf{X}_2 + \cf{Z}_1$};
            \draw[fwdrxn] (a)--(b) node [midway, above] {\ratecnst{$\kk_{21}$}}; 
        \end{scope}
    \end{scope} 
    \node at (0,-1.5) {};
    \end{tikzpicture}
   \end{split} 
    \end{align}
where $\cf{Z}_i$ is the inactive form of $\cf{X}_i$. 

Under mass-action kinetics, the concentrations evolve according to 
\eq{ 
    -\dot{z}_1 = \dot{x}_1 &= k_\text{on}^1 z_1 - k_\text{off}^1 x_1 - k_{21} x_1x_2 + k_c C z_1
    \\
    -\dot{z}_2 = \dot{x}_2 &= k_\text{on}^2 z_2 - k_\text{off}^2 x_2 + k_{32} x_3 z_2
    \\
    -\dot{z}_3 = \dot{x}_3 &= k_\text{on}^3 z_3 - k_\text{off}^3 x_3 + k_{13} x_1 z_3.
}
Since $z_i + x_i = T_i$ for some total concentration $T_i > 0$, we can eliminate $z_i$. The reduced system  is 
\eq{ 
    \dot{x}_1 &= k_\text{on}^1 (T_1 - x_1) - k_\text{off}^1 x_1 - k_{21} x_1x_2 + k_c C (T_1 - x_1)
    \\
    \dot{x}_2 &= k_\text{on}^2 (T_2 - x_2) - k_\text{off}^2 x_2 + k_{32} x_3 (T_2 - x_2)
    \\
    \dot{x}_3 &= k_\text{on}^3 (T_3 - x_3) - k_\text{off}^3 x_3 + k_{13} x_1 (T_3 - x_3),
}
where $0 < x_i < T_i$, since the positive orthant is forward-invariant under mass-action kinetics~\citeSM[Lemma II.1]{Sontag2001b}. Its Jacobian matrix is 
\eq{ 
    \mm J = \begin{pmatrix}
        -k_\text{on}^1 - k_\text{off}^1 - k_{21}x_2 -  k_c C & -k_{21} x_1  & 0 \\ 
        0 & -k_\text{on}^2 - k_\text{off}^2 - k_{32} x_3 
        & k_{32} (T_2 - x_2) \\ 
        k_{13}(T_3 - x_3) & 0 & -k_\text{on}^3 - k_\text{off}^3 - k_{13} x_1 
    \end{pmatrix}. 
}
Since $(T_i - x_i) > 0$, the J-graph in \Cref{fig:activation-inhibition-network} (without the colours and texts) can be interpreted as the J-graph of the reduced system.

\bigskip\bigskip

\newpage 
\subsection{\sc Details of \Cref{ex:biorxiv2}} 
\label{sec:app-ex:bioarxiv2}

This example first appeared in \citeSM{Sontag2020_biphasic_biorxivb}, which also contains  its derivation. 

A small technical issue is that the saturation $\sigma $ in the system \eqref{eq:ex:biorxiv2}, while a (globally) Lipschitz function (insuring unique and everywhere defined solutions of the ODE) is not differentiable at those states for which either $x_1$ or $x_2$ equals exactly $0.8$ or $1.2$, which happens only for a set of measure zero.  Thus the Jacobian is not well-defined at every state in the classical sense.  However, it is well-defined in the sense of nonsmooth analysis, and in any event the signs of interactions reflect the increasing or decreasing influence of each variable on every other variable, and this can be defined with no need to take derivatives.  The saturation function $\sigma $ could be replaced by an approximating differentiable function, but use of $\sigma $ makes calculations somewhat simpler.

\Cref{fig:app-bioarxiv2-ss} shows the steady state response curves for all variables, as well as the eigenvalues along the steady state curve, indicating asymptotic stability. \Cref{fig:app-bioarxiv2-ss-z_w} shows, in particular, the stable biphasic response of $x_3$ and $x_4$.  Global stability was indicated by simulations with random initial conditions. 

\begin{figure}[ht]
\centering 
\begin{subfigure}[t]{0.315\textwidth}
    \centering 
    \includegraphics[height=1.35in]{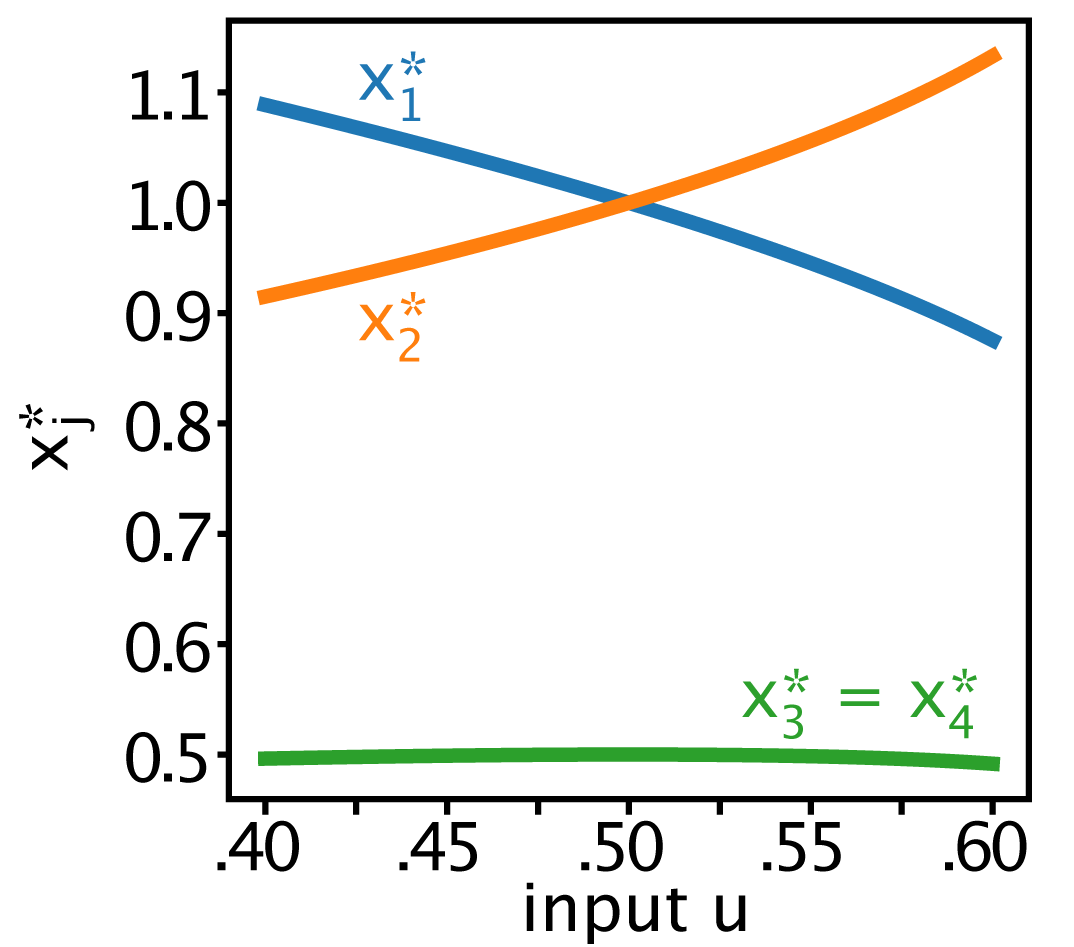} 
    \caption{}
    \label{fig:app-bioarxiv2-ss-all}
\end{subfigure}
\begin{subfigure}[t]{0.34\textwidth}
    \centering 
    \includegraphics[height=1.35in]{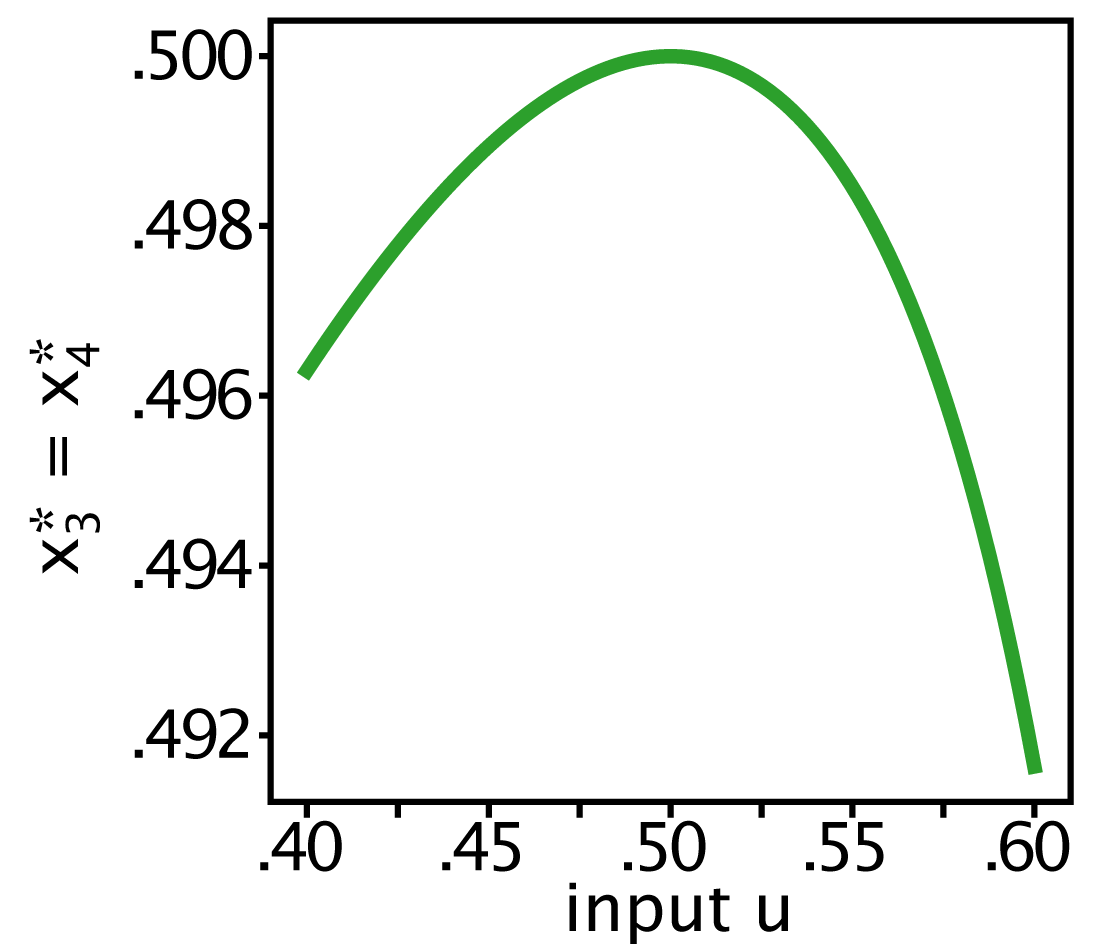}\hspace{0.25cm} 
    \caption{}
    \label{fig:app-bioarxiv2-ss-z_w}
\end{subfigure}
\begin{subfigure}[t]{0.325\textwidth}
    \centering 
    \includegraphics[height=1.35in]{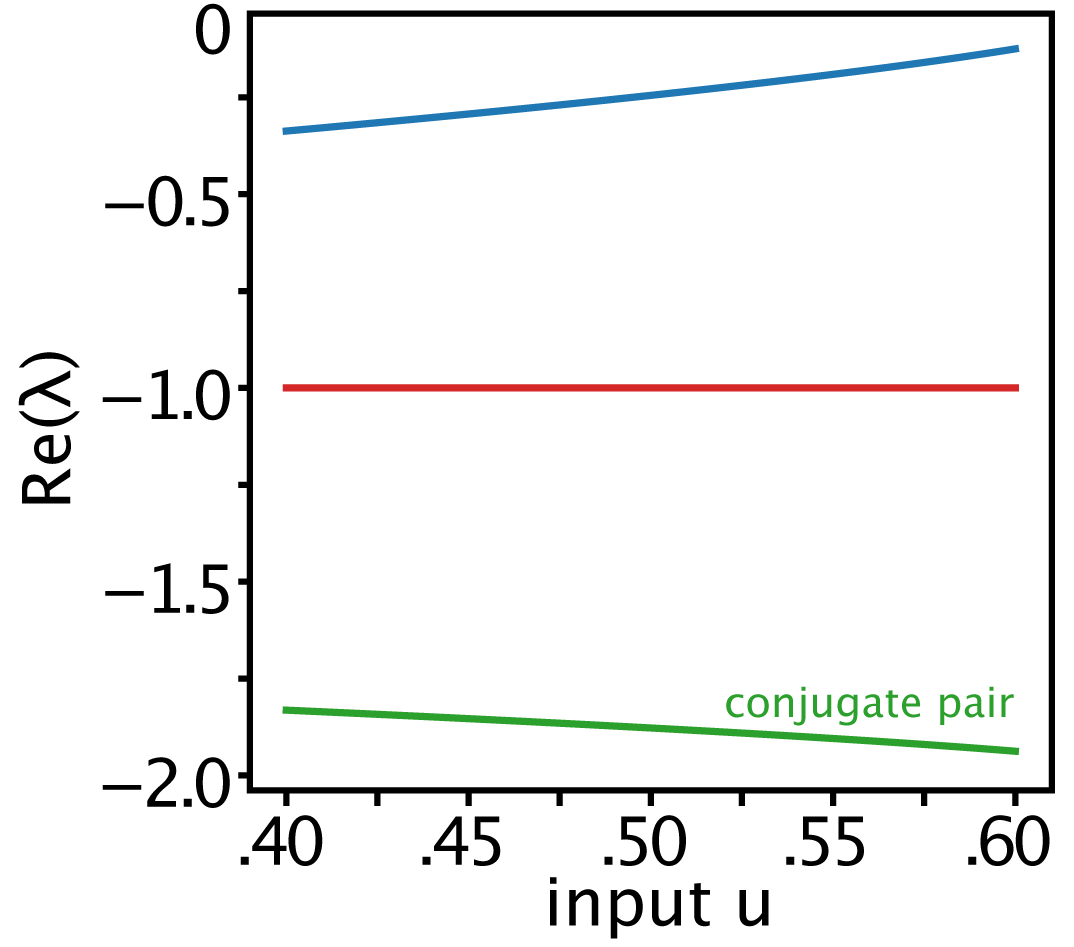} \hspace{0.1cm} 
    \caption{}
    \label{fig:app-bioarxiv2-ss-evals}
\end{subfigure}
\caption{
   (a) Steady state curves of \Cref{ex:biorxiv2}. (b) $x_3, x_4$ exhibit biphasic responses where $x_3^* = x_4^*$. (c) Real parts of eigenvalues at steady state, demonstrating that the steady state response is stable. 
}
\label{fig:app-bioarxiv2-ss}
\end{figure}

\bigskip
\subsubsection{Existence and uniqueness of steady state}

We show that for each $u \in [0.4,0.6]$, there exists a unique positive steady state. For example, one can easily check that for $u_0 = 0.5$, the point $(1,1,0.5,0.5)$ is a steady state for \Cref{ex:biorxiv2}.  More generally, from $\dot x_4 = 0$, we know that $x_4 = x_3$ at steady state. From $\dot x_1 = 0$, we have 
\begin{align}
\label{eq:app-biorxiv2-x1}
    x_1 = e^{1-\sigma(x_2)}. 
\end{align}
Solving $\dot x_3=0$ for $x_3$, and using \eqref{eq:app-biorxiv2-x1} so that $\left(e^{(1-\sigma(x_2))}\right)^2 =x_1^2$, we obtain
\begin{align}
\label{eq:app-biorxiv2-x3x4}
    x_3 = x_4 =  \frac{x_1^2}{4} + \frac{u}{2} .
\end{align}
Finally, we set $\dot x_2=0$, substituting the above $x_3$, to obtain an equation relating
$x_1$, $x_2$, and $u$:
\begin{align*}
    \frac{e^{2(1-\sigma(x_1))}}{2} - x_2 +  \frac{u}{2}  + \frac{x_1^2}{4}  = 0 . 
\end{align*}
This equation is equivalent to
\begin{align}
\label{eq:app-biorxiv2-U}
    u = U(x_1,x_2) \coloneqq  -\frac{x_1^2}{2} + 2x_2 - e^{2(1-\sigma(x_1))} .
\end{align}
Now we will show that, for each $u\in [0.4,0.6]$, there is a unique solution $(x_1,x_2)$ for $U(x_1,x_2)=u$, and moreover, $x_1\in [0.8,1.2]$ and $x_2\in [0.8,1.2]$, i.e., in the nonsaturation regime, case (1A) below. We do this by analyzing various cases.

\medskip 
\noindent\textbf{Case 1: $0.8 \leq x_2 \leq 1.2$.}
In this case, \eqref{eq:app-biorxiv2-x1} gives $x_2 = 1 - \ln x_1$, so 
\begin{align}
\label{eq:app-biorxiv2-U-case1}
    U(x_1,x_2) = -\frac{x_1^2}{2} +2 - 2\ln x_1 - e^{2(1-\sigma(x_1))} .
\end{align}
For $0.8 \leq x_1 \leq 1.2$, the above expression is  $U(x_1, 1-\ln x_1) = -\frac{x_1^2}{2} +2 - 2\ln x_1 - e^{2(1-x_1)}$, a plot of which is shown in \Cref{fig:app-bioarxiv2-pf-ss_existence-Case1}, middle, blue segment. This is a decreasing function on the interval $x_1 \in [0.8,1.2]$. For each $u \in [0.4,0.6]$, there exists a unique $x_1 \in [0.8,1.2]$ solving \eqref{eq:app-biorxiv2-U}; see the shaded blue band in \Cref{fig:app-bioarxiv2-pf-ss_existence-Case1}. This also returns $x_2 = 1 - \ln x_1$, and $x_3,x_4$ via \eqref{eq:app-biorxiv2-x3x4}. 

If instead, $x_1 < 0.8$, then $U(x_1,1-\ln x_1) = -\frac{x_1^2}{2} +2 - 2\ln x_1- e^{0.4}$, which is decreasing on $x_1 \in [0,0.8]$ with minimum value of $U \approx 0.63$ (\Cref{fig:app-bioarxiv2-pf-ss_existence-Case1}, left, orange segment). Finally, if $x_1 > 1.2$, then $U(x_1,1-\ln x_1) = -\frac{x_1^2}{2} +2 - 2\ln x_1- e^{-0.4}$, which is decreasing on $x_1 \in [1.2, \infty)$, with maximum value of $U \approx 0.245$ (\Cref{fig:app-bioarxiv2-pf-ss_existence-Case1}, right, green segment). So \eqref{eq:app-biorxiv2-U} has no corresponding solution for which $x_1 < 0.8$ or $x_1 > 1.2$.

\begin{figure}[ht]
\centering 
\begin{subfigure}[t]{0.315\textwidth}
    \centering 
    \includegraphics[height=1.35in]{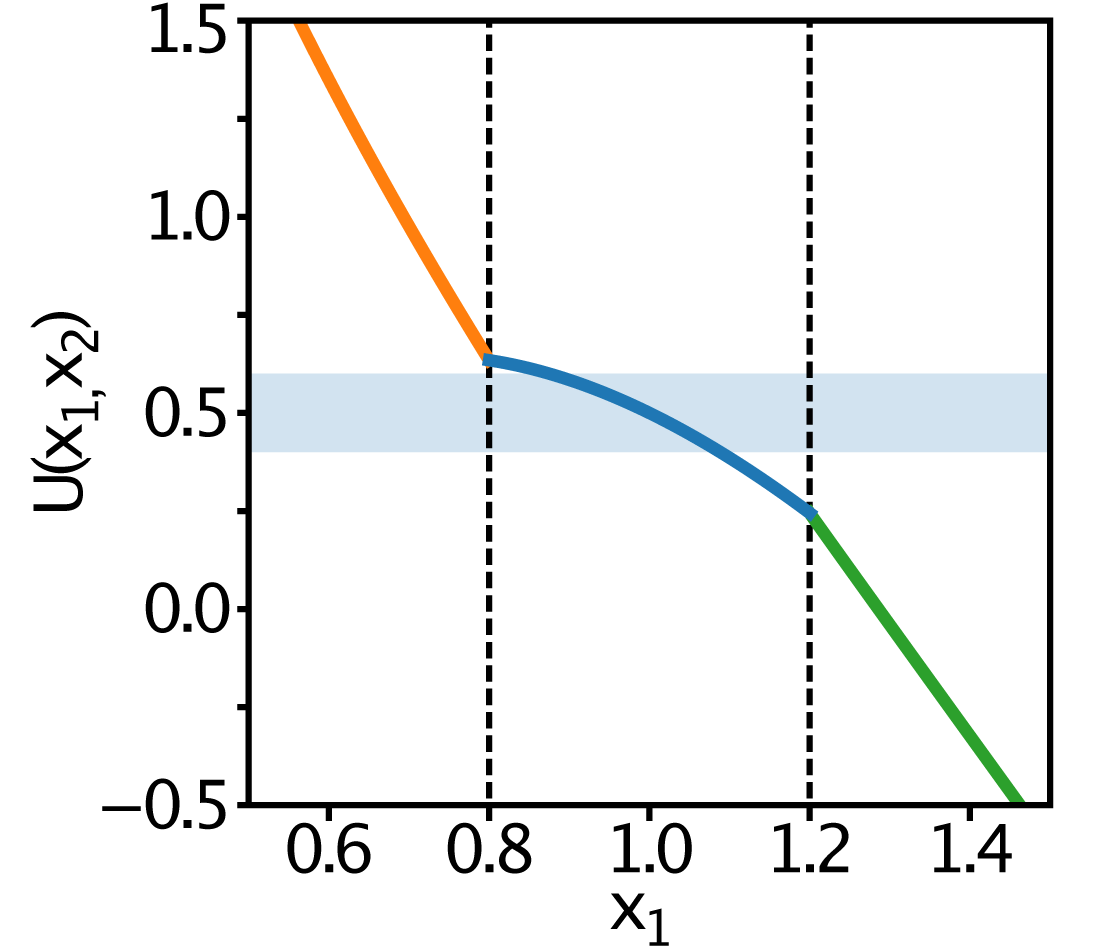}
    \caption{}
    \label{fig:app-bioarxiv2-pf-ss_existence-Case1}
\end{subfigure}
\begin{subfigure}[t]{0.34\textwidth}
    \centering 
    \includegraphics[height=1.35in]{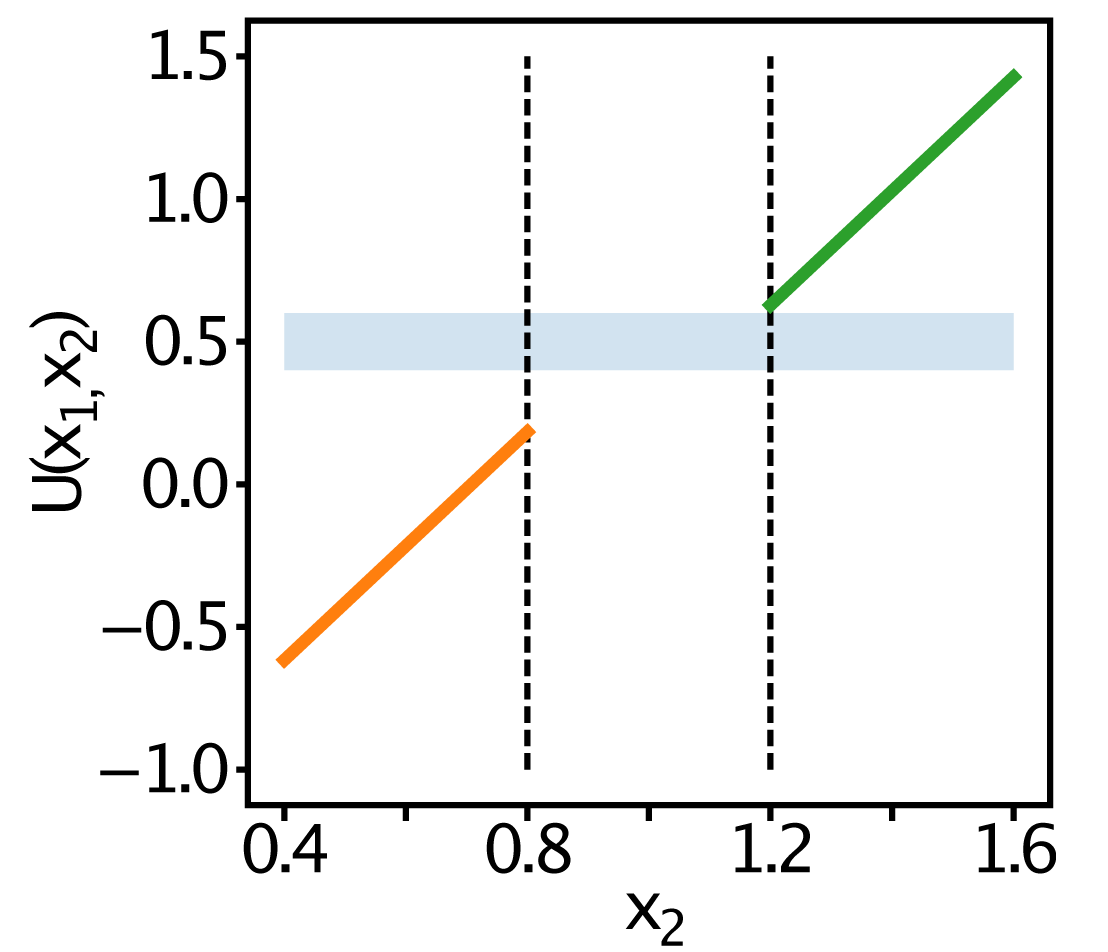}
    \caption{}
    \label{fig:app-bioarxiv2-pf-ss_existence-Case2_3}
\end{subfigure}
\caption{
   Graph of \eqref{eq:app-biorxiv2-U} in different cases. (a) $U(x_1,x_2)$ when $0.8 \leq x_2 \leq 1.2$ (case 1), as given by \eqref{eq:app-biorxiv2-U-case1}. (b) $U(x_1,x_2)$ when $x_2 < 0.8$ (case 2) and when $x_2 > 1.2$ (case 3). Shaded region is the range $[0.4,0.6]$ in which $u$ lies. Solution to $u = U(x_1,x_2)$ exists only when $0.8 \leq x_1, x_2 \leq 1.2$. 
}
\label{fig:app-bioarxiv2-pf-ss_existence}
\end{figure}

\medskip 
\noindent\textbf{Case 2: $x_2 < 0.8$.}
In this case, \eqref{eq:app-biorxiv2-x1} gives $x_1 = e^{0.2} \approx 1.22$, and $\sigma(x_1) = 1.2$. Then \eqref{eq:app-biorxiv2-U} with 
\eq{ 
    U(x_1,x_2) =  -\frac{e^{0.4}}{2} + 2x_2 - e^{-0.4}  
}
has no solution to \eqref{eq:app-biorxiv2-U} for $0.4 \leq u \leq 0.6$ when  $x_2 < 0.8$, since $U(x_1,x_2) < 0.19$ (\Cref{fig:app-bioarxiv2-pf-ss_existence-Case2_3}, left, orange segment). 

\medskip 
\noindent\textbf{Case 3: $1.2 < x_2$.} In this case, \eqref{eq:app-biorxiv2-x1} gives $x_1 = e^{-0.2} \approx 0.81$, and $\sigma(x_1) = x_1$. Then \eqref{eq:app-biorxiv2-U} with 
\eq{ 
    U(x_1,x_2) =  -\frac{e^{-0.4}}{2} + 2x_2 - e^{2(1-e^{-0.2})}  
}
has no solution to \eqref{eq:app-biorxiv2-U} for $0.4 \leq u \leq 0.6$ when $x_2 >1.2 $, since $U(x_1,x_2) > 0.62$ (\Cref{fig:app-bioarxiv2-pf-ss_existence-Case2_3}, right, green segment).

We conclude that steady states are unique, and occur only when $0.8 \leq x_1, x_2 \leq 1.2$, when neither variable is saturated. In this case, we solve \eqref{eq:app-biorxiv2-U}  numerically to get $x_1=x_1^*(u)$ for each constant input $u$ in the interval $[0.4,0.6]$ and back substitute to obtain the remaining variables, obtaining the steady state curve $\vv x^*(u)$ in \Cref{fig:app-bioarxiv2-ss-all}.

\bigskip
\subsubsection{Stability of steady state response curve}
It is not difficult to see that, along the steady state curve, the Jacobian matrix of \eqref{eq:ex:biorxiv2} is 
\begin{align*}
    \mm J = \begin{pmatrix}
        -1 & - C & 0  & 0\\ 
        -e^{2(1 - C)}  & -1 & 1 & 0 \\ 
        0 & -C^2 & -2  & 0 \\ 
        0 & 0 & 0 & -1 
    \end{pmatrix}, 
\end{align*}
where $C \coloneqq e^{1-x_2^*} > 0$. Its J-graph is the one in \Cref{fig:biorxiv2-Jgraph}. 
The real part of the eigenvalues of $\mm J$ (\Cref{fig:app-bioarxiv2-ss-evals})  is always negative, showing asymptotic stability. Simulations from random initial states indicate that the state $\vv x^*(u)$ for $u\in [0.4,0.6]$
is globally asymptotically stable.

\bigskip\bigskip 
\newpage 
\subsection{\sc Details of \Cref{ex:biphasic_unstable}} 
\label{sec:app-ex:biphasic_unstable}

\Cref{ex:biphasic_unstable} is generated by the reactions 
\begin{align*}
\begin{split}
    \begin{tikzpicture}[scale=1]
        \node (0) at (0,0) {$\cf{0}$};
        \node (1) at (1.75,0) {$\cf{X}_1$};
        \node (4) at (1.75,-1.5)  {$\cf{X}_4$};
        \node (2) at (3.5,0) {$\cf{X}_2$};
        \node (3) at (3.5,-1.5) {$\cf{X}_3$};
        \node (22) at (5.5,-1.5){$2\cf{X}_2$};
        \draw[fwdrxn] (0)--(1) node[midway, above] {\ratecnst{{$u$}}}; 
        \draw[fwdrxn] (1)--(4) node[midway, left] {\ratecnst{$1$}};
        \draw[fwdrxn] (1)--(2) node[midway, above] {\ratecnst{$1$}};
        \draw[fwdrxn] (3)--(2) node[midway, right] {\ratecnst{$1$}};
        \draw[fwdrxn] (3)--(4) node[midway, above] {\ratecnst{$3/4$}};
        \draw[fwdrxn] (4) --(0)  node [midway, below left] {\ratecnst{$1$}\!\!};
        \draw[fwdrxn] (22)--(3) node[midway, above] {\ratecnst{$1$}};
        \begin{scope}[shift={(-4.3,-2.75)}]
        \node (33) at (6,0) {$2\cf{X}_3$};
        \node (332) at (7.5,0) [right] {$2\cf{X}_3 + \cf{X}_2$};
        \draw[fwdrxn] (33)--(332) node[midway, above] {\ratecnst{$1/4$\,}};
        \end{scope}
    \node at (0, -3.25) {};
    \end{tikzpicture}
\end{split}
\end{align*} 
which under mass-action kinetics, is governed by \eqref{eq:ex:biphasic_unstable}. Its Jacobian matrix is 
    \eq{ 
        \mm J = \begin{pmatrix}
            -2 & 0 & 0 & 0 \\ 
            1 & -4x_2 & 1 +  x_3/2 & 0 \\ 
            0 & 2x_2 & -7/4 & 0 \\ 
            1 & 0 & 3/4 & -1 
        \end{pmatrix}.
    }
We can solve for the nonnegative steady states, or check that 
    \eqn{
    \begin{split}
    \label{eq:unstable_biphasice-ss_curves}
        x_1^*(u) &= \frac{u}{2} , \\ 
        x_2^*(u) &= \frac{\sqrt 7}{2} \sqrt{5 \pm \sqrt{25-2u}}, \\ 
        x_3^*(u) &= 5 \pm \sqrt{25-2u} , \\
        x_4^*(u) &
            = \frac{u}{2} + \frac{15}{4} \pm \frac{3}{4} \sqrt{25 - 2u}
    \end{split}
    }
are two branches of steady states. 

Using Mathematica or other tools, we can compute the eigenvalues along the steady state curves. \Cref{fig:unstable_biphasic-evals} are the plots of the eigenvalues.

\begin{figure}[ht]
\centering 
\begin{subfigure}[t]{0.35\textwidth}
    \centerline{ 
    \includegraphics[height=1.35in]{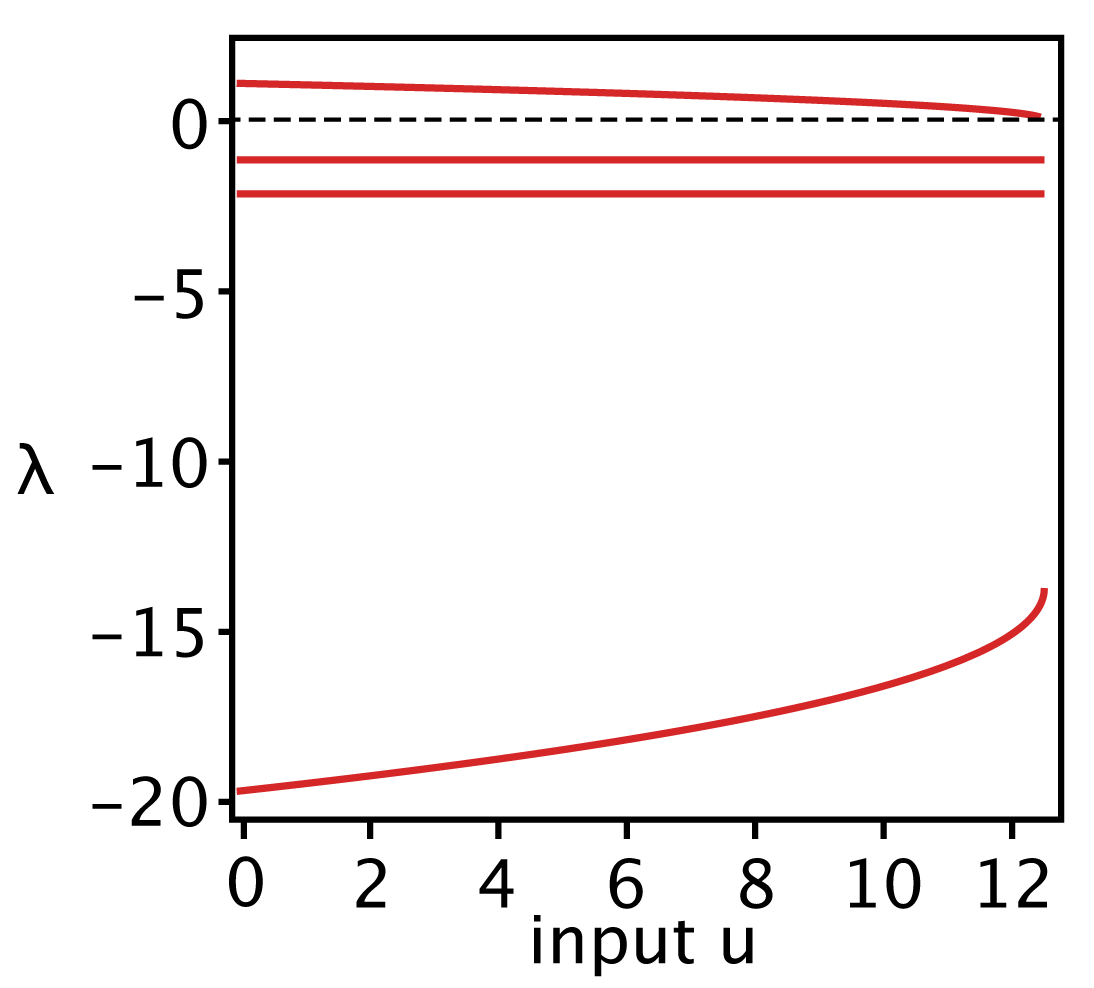} 
    }
    \caption{}
    \label{fig:unstable_biphasic-evals-a}
\end{subfigure}
\begin{subfigure}[t]{0.35\textwidth}
    \centerline{ 
    \includegraphics[height=1.35in]{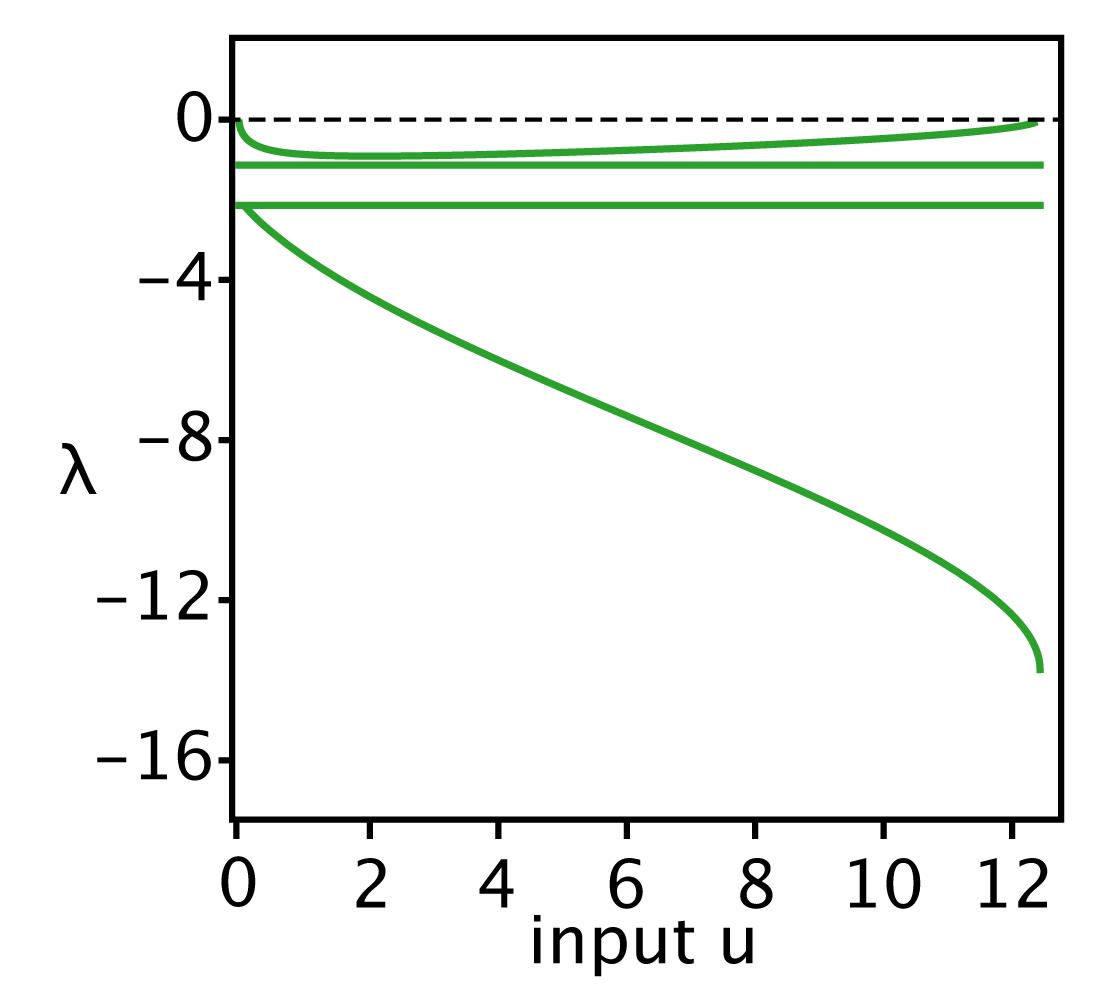} 
    }
    \caption{}
    \label{fig:unstable_biphasic-evals-b}
\end{subfigure}
\caption{
    Eigenvalues along the steady state curves \eqref{eq:unstable_biphasice-ss_curves}. (a) Corresponding to the unstable branch in \Cref{fig:biphasic_unstable}, and (b) the stable branch. 
}
\label{fig:unstable_biphasic-evals}
\end{figure}

\bigskip
\subsubsection{Derivation of steady state}

For completeness, we derive the expression \eqref{eq:unstable_biphasice-ss_curves} for the steady state curve  for the system \eqref{eq:ex:biphasic_unstable}.  Substituting $x_1 = u/2$ from $\dot x_1 = 0$ and  $x_2^2 = 7x_3 /4$ from  $\dot x_3 = 0$ into $\dot x_2 = 0$, we see that 
\eq{ 
    \frac{u}{2} + x_3 - \frac{7}{2} x_3 + \frac{x_3^2}{4} = 0, 
}
which we rescale and then solve using the quadratic formula 
\eq{ 
    x_3 = \frac{10 \pm \sqrt{100 - 8u}}{2}
        = 5 \pm \sqrt{25 - 2u}.
}
Then 
\eq{ 
    x_2 = \frac{\sqrt{7}}{2} \sqrt{5 \pm \sqrt{25 - 2u}}.
}
Finally from $\dot x_4 = 0$, we have 
\eq{ 
    x_4 = \frac{u}{2} + \frac{15}{4} \pm \frac{3}{4} \sqrt{25-2u}.
}

\bigskip
\subsubsection{Expressions for the eigenvalues along steady state curve}

Plotted in \Cref{fig:unstable_biphasic-evals-a} are the eigenvalues of the (unstable) branch $x_4^*(u) = \frac{3}{4} \left( \frac{2u}{3} + 5 + \sqrt{25-2u}\right)$:
    \eq{ 
        \lambda_1 & =-2 \\ 
        \lambda_2 &= -1 \\ 
        \lambda_3 &= \frac{1}{8} \left( 
            -\sqrt{49-48\sqrt 7 \sqrt{5 + \sqrt{25 - 2u} }
                + 448 (5 + \sqrt{25-2u}) 
                + 32 \sqrt 7 (5 + \sqrt{25-2u})^{3/2}
            }
            \right. \\&\qquad\qquad \left. \vphantom{\sqrt{\sqrt{5+\sqrt{25-2u}}}}
            - 8\sqrt 7\sqrt{5 + \sqrt{25-2u}} -7 
            \right)  \\ 
        \lambda_4 &= \frac{1}{8} \left( 
            \sqrt{49-48\sqrt 7 \sqrt{5 + \sqrt{25 - 2u} }
                + 448 (5 + \sqrt{25-2u}) 
                + 32 \sqrt 7 (5 + \sqrt{25-2u})^{3/2}
            }
            \right. \\&\qquad\qquad \left. \vphantom{\sqrt{\sqrt{5+\sqrt{25-2u}}}}
            - 8\sqrt 7\sqrt{5 + \sqrt{25-2u}} -7 
            \right) .
    }
Plotted in \Cref{fig:unstable_biphasic-evals-b} are the eigenvalues of the (stable) branch $x_4^*(u) = \frac{3}{4} \left( \frac{2u}{3} + 5 - \sqrt{25-2u}\right)$: 
    \eq{ 
        \lambda_1 & =-2 \\ 
        \lambda_2 &= -1 \\ 
        \lambda_3 &= \frac{1}{8} \left( 
            -\sqrt{49-48\sqrt 7 \sqrt{5 - \sqrt{25 - 2u} }
                + 448 (5 - \sqrt{25-2u}) 
                + 32 \sqrt 7 (5 - \sqrt{25-2u})^{3/2}
            }
            \right. \\&\qquad\qquad \left. \vphantom{\sqrt{\sqrt{5+\sqrt{25-2u}}}}
            - 8\sqrt 7\sqrt{5 - \sqrt{25-2u}} -7 
            \right)  \\ 
        \lambda_4 &= \frac{1}{8} \left( 
            \sqrt{49-48\sqrt 7 \sqrt{5 - \sqrt{25 - 2u} }
                + 448 (5 - \sqrt{25-2u}) 
                + 32 \sqrt 7 (5 - \sqrt{25-2u})^{3/2}
            }
            \right. \\&\qquad\qquad \left. \vphantom{\sqrt{\sqrt{5+\sqrt{25-2u}}}}
            - 8\sqrt 7\sqrt{5 - \sqrt{25-2u}}-7 
            \right) .
    }
These expressions were obtained using Mathematica.

\bigskip \bigskip 
\newpage 
\subsection{\sc Dynamics of the core kinetic proofreading network}
\label{sec:app-KP_core}

Under mass-action kinetics, the system governing the dynamics of the core network (\Cref{fig:KP-core}) is 
\begin{align}
\begin{split}
\label{eq:app-KP_core_ODE}
    \frac{dL}{dt} &= -\kk LR + \ell_0 C_0 + \ell_1 C_1 + \cdots + \ell_N C_N \\ 
    \frac{dR}{dt} &= -\kk LR + \ell_0 C_0 + \ell_1 C_1 + \cdots + \ell_N C_N \\ 
    \frac{dC_0}{dt} &= \kk LR - (\kk_0 + \ell_0) C_0 \\ 
    & \qquad \vdots \\ 
    \frac{dC_i}{dt} &= \kk_{i-1} C_{i-1} - (\kk_i + \ell_i) C_i , \quad \text{ for } i = 1,2,\ldots, N-1\\ 
    & \qquad \vdots \\ 
    \frac{dC_N}{dt} &= \kk_{N-1} C_{N-1} - \ell_N C_N.
\end{split}
\end{align}
Let $C_T \coloneqq \sum_{i=0}^N C_i$,  $L_T = L + C_T$, and $R_T = R + C_T$. Note that $C_T$, $L$, and $R$ evolve in time, while $L_T, R_T$ are constants indicating the total amount of ligands and receptors, whether in free form or bound form.

\bigskip
\subsubsection{Analytical expression for the steady states of TCR complexes}
\label{sec:app-application-KP_ss}

We claimed in \Cref{sec:application-KP_ss} that $C_i^* = A_i C_T^*$ and 
\begin{align*}
    C_T^* = \frac{L_T + R_T + \frac{A}{\kk}}{2} - \sqrt{ \left( \frac{L_T + R_T + \frac{A}{\kk}}{2}\right)^2 - L_T R_T} , 
\end{align*}
where the constants $L_T, R_T, A,\kk$ are positive. We derive the above expression by following the idea in \citeSM{LeverLimKrugerNguyenEtAl2016b}.

Suppose the system is at steady state. We drop the asterisk to simplify notation. We have 
\eq{ 
    C_{N-1} &= \frac{\ell_N}{\kk_{N-1}} C_N, \\ 
    C_{i-1} &= \frac{\kk_i + \ell_i}{\kk_{i-1}} C_i, \quad \text{ for } i = 2, \ldots, N-1, \\ 
    C_0 &= \frac{\kk_1 + \ell_1}{\kk_0} C_1 ,\\ 
    \kk LR &= (\ell_0 + \kk_0) C_0 = \sum_{j=0}^N \ell_j C_j . 
}
Note that the equations involving $C_i$ can be calculated recursively:
\eq{ 
    &\quad C_0 = \frac{k_1 + \ell_1}{k_0} C_1 , \quad C_1 = \frac{k_2 + \ell_2}{k_1} C_2 , \quad \cdots \\ 
    \implies & \quad C_0 = \frac{k_1 + \ell_1}{k_0}\frac{k_2 + \ell_2}{k_1} C_2 , \quad \cdots 
}
resulting in the relations 
\eq{ 
    C_{N-1} &= \frac{\ell_N }{\kk_{N-1}} C_N , \\ 
    C_i &= \frac{\ell_N \prod_{j=i+1}^{N-1} (k_j + \ell_j) }{\prod_{j=i}^{N-1} k_j} C_N, \quad \text{ for } i=1,2,\ldots, N-2 \\
    C_1 &= \frac{\ell_N \prod_{j=2}^{N-1} (\kk_j + \ell_j)}{\prod_{j=1}^{N-1} \kk_j } C_N 
        = \frac{(\kk_2 + \ell_2)(\kk_3+\ell_3) \cdots  (\kk_{N-1}+\ell_{N-1}) \ell_N}{ \kk_1 \kk_2 \cdots \kk_{N-1}} C_N \\ 
    C_0 &= \frac{(\kk_1 + \ell_1)(\kk_2+\ell_2) \cdots (\kk_{N-1}+\ell_{N-1}) \ell_N   }{ \kk_0 \kk_1 \kk_2 \cdots \kk_{N-1}} C_N . 
}
For simplicity, we will denote the above as $C_i = a_i C_N$, where $a_i > 0$ is the appropriate product. 

Now consider the definition of $C_T$, and substitute for $C_j = a_j C_N$:
\eq{ 
    C_T &= C_0 + C_1 + \cdots + C_N  
    = (a_0 + a_1 + \cdots +  a_{N-1} + a_N) C_N, 
}
where $a_N = 1$. This means that 
\eq{ 
    C_N &= \frac{1}{\sum_{j=0}^{N} a_j} C_T \eqqcolon A_N C_T, \\ 
    C_i &= \frac{a_i}{\sum_{j=0}^{N} a_j}C_T \eqqcolon A_i C_T , \quad \forall i =0,1,2,\ldots, N-1. 
} 

Next, we have the conservation laws $L_T = L + C_T$, and $R_T = R + C_T$, so 
\eq{ 
    LR = (L_T - C_T )(R_T - C_T) = C_T^2 - C_T(L_T + R_T)  + L_TR_T. 
}
On one hand, from $\frac{dC_0}{dt} = 0$, we see that 
\eq{
    kLR = (k_0 + \ell_0) C_0 = A C_T,
}
where $A = (k_0 + \ell_0) A_0$. 
On the other hand, this is equal to $\kk (C_T^2 - C_T(L_T + R_T)  + L_TR_T)$. Thus, 
\eq{ 
     C_T^2 -  C_T(L_T + R_T + A/\kk )  +  L_TR_T = 0 .
}
By the quadratic formula, 
\eqn{ \label{eq:core-CT}
    C_T &= \frac{L_T + R_T + A/\kk}{2} - \sqrt{ \left( \frac{L_T + R_T + A/\kk}{2}\right)^2 - L_TR_T}, 
}
where the positive root is rejected, because $C_T \leq \min \{ L_T, R_T\}$. This last claim follows from 
\eq{ 
    \left( \frac{L_T + R_T + A/\kk}{2}\right)^2 - L_TR_T
    &= \left( \frac{L_T+R_T}{2} \right)^2 - L_TR_T + \frac{A^2}{4k^2} + \frac{A(L_T+R_T)}{k} 
    >0 
}
by the AM-GM inequality.

\bigskip
\subsubsection{Monotonic dependence on total antigen concentration at steady state}
\label{sec:app-CT_monotone}

To show that $\frac{\partial C_T}{\partial L_T} > 0$ for all $L_T > 0$, it suffices to show that the function $g(x) = x+a+b - \sqrt{(x+a+b)^2 - 4ax}$ has positive derivative for all $x > 0$. This is a simple calculation, as $g'(x) = 1 - \frac{(x+a+b) - 2a}{\sqrt{(x+a+b)^2 - 4ax} }$. Thus, $g'(x) > 0$ if and only if 
\eq{
\begin{array}{crcl}
    &\qquad \sqrt{(x+a+b)^2 - 4ax}  &>& (x+a+b) - 2a \\ 
\iff & \qquad  (x+a+b)^2 - 4ax &>& (x-a+b)^2 \\ 
\iff &\qquad  (x+b)^2 + a^2 + 2a(x+b) - 4ax &>& (x+b)^2 + a^2 - 2a(x+b) \\ 
\iff &\qquad -2ax + 2ab &>& -2ax -2ab, 
\end{array}
}
which is clearly true. Thus $\frac{\partial C_T}{\partial L_T} > 0$ holds.

\bigskip \bigskip 
\newpage 
\subsection{\sc System with multiple inputs}
\label{sec:app-multiple_inputs}

Here we give in detail how to extend our results to a system with multiple inputs. We take an abstract system with $3$ variables, in which $2$ of them are influenced by $u$. Consider the system 
\begin{align}
\begin{split}
\label{eq:app-multiple_inputs}
    \dot x_1 &= f_1(\vv x) + u \\ 
    \dot x_2 &= f_2(\vv x) + u \\ 
    \dot x_3 &= f_3(\vv x) , 
\end{split}
\end{align}
where we assume $\frac{\partial f_i}{\partial x_i} < 0$ for each $i$ and every $\frac{\partial f_j}{\partial x_i}$ is constant in sign. We define an extended system by treating $u$ as a variable $x_0$, resulting in the system with control $v$:   
\begin{align}
\begin{split}
\label{eq:app-multiple_inputs-extended}
    \dot x_0 &= -x_0 + v, \qquad x_0(0) = v \\ 
    \dot x_1 &= f_1(\vv x) + x_0 \\ 
    \dot x_2 &= f_2(\vv x) + x_0 \\ 
    \dot x_3 &= f_3(\vv x) . 
\end{split}
\end{align}

Then along a steady state curve $(x_0(v), \xx^*(v))$, which we assume it exists, we can implicitly differentiate with respect to $v$ and obtain 
\eq{ 
    \begin{pmatrix} -1 \\ 0 \\ 0 \\ 0 \end{pmatrix}
    = \begin{pmatrix} -1 & 0 & 0 & 0 \\ 
        1 & \partial_1 f_1 & \partial_2 f_1 & \partial_3 f_1 \\ 
        1 & \partial_1 f_2 & \partial_2 f_2 & \partial_3 f_2 \\ 
        0 & \partial_1 f_3 & \partial_2 f_3 & \partial_3 f_3
    \end{pmatrix}
    \begin{pmatrix} \partial_v x^*_0 \\ \partial_v x^*_1 \\ \partial_v x^*_2 \\ \partial_v x^*_3 \end{pmatrix}. 
}
Let $\mm J$ be the matrix above. Then Cramer's rule tells us that for $j = 1,2,3$, 
\eq{ 
    \partial_v x_j^* = \frac{(-1)^{j+1} \mm J[\hat 1, \widehat{j+1}]}{\det \mm J}. 
}
Note that $\mm J$ is \emph{not} $\frac{\partial \vv f}{\partial \vv x}$, but contains it as a submatrix.

From \eqref{eq:app-multiple_inputs-extended}, clearly $x_0(t) = v$ for all $t \geq 0$. So for $j=1,2,3$, the derivative of $x_j^*$ with respect to $v$ in \eqref{eq:app-multiple_inputs-extended} is equal to the derivative of $x_j^*$ with respect to $u$ in \eqref{eq:app-multiple_inputs}, where the values of $u$ and $v$ are chosen to be the same constant.

\bigskip
\subsubsection{An example with multiple inputs}
\label{sec:app-multiple_inputs-concrete_ex}

As a concrete example, consider the system with output $z$
\eqn{
\begin{split}
\label{eq:app-mult_inputs_concrete_original}
    \dot x &= -x + u \\ 
    \dot y &= -y + u \\ 
    \dot z &= -z + x - y^2 z .
\end{split}
}
The J-graph of this system is shown in \Cref{fig:mutliple_inputs_ex-Jgraph}, since its Jacobian matrix is 
\eq{ 
    \begin{pmatrix} 
        -1 & 0 & 0 \\ 
        0 & -1 & 0 \\ 
        1 & -2yz & -1 -y^2 
    \end{pmatrix} .
}
This system is stable biphasic in $z$; see \Cref{fig:mutliple_inputs_ex-response}. At first sight, it seems as though our results are silent on this system; however, we can extend our results by treating $u$ as a variable.

Consider the following system with control $v$ instead:
\eqn{
\begin{split}
\label{eq:app-mult_inputs_concrete_extended}
    \dot x_0 &= -x_0 + v, \qquad x_0(0) = v \\ 
    \dot x &= -x + x_0 \\ 
    \dot y &= -y + x_0 \\ 
    \dot z &= -z + x - y^2 z .
\end{split}
}
The J-graph of this extended input-output system is \Cref{fig:mutliple_inputs_ex-Jgraph-extended}. \Cref{cor:summary} applies; that $z^*(v)$, where $u=v$, is stable biphasic implies the existence of an IFFL, or both types of FBLs. Note the presence of an IFFL from $x_0$ (formally the control $u$) to the output $z$. Thus, to apply our results, we need to include in the J-graph a node for any control that influences more than one variable. 

\begin{figure}[h]
\centering 
\begin{subfigure}[t]{0.28\textwidth}
    \centering
    \begin{tikzpicture}
            \node[Jnode] (1) at (0,0) {$x$};  
            \node[Jnode] (2) at (0,-1.5) {$y$}; 
            \node[Jnode] (3) at (1.25,-0.75) {$z$}; 
            \draw[posEdge] (1)--(3);
            \draw[negEdge] (2)--(3);
            \draw[posEdge] (3)--(2.25,-0.75);
            \draw[posEdge] (-1.2,0)--(1);
            \draw[posEdge] (-1.2,-1.5)--(2); 
        \node at (-2.25,-0.75){};
        \node at (0,-2.7) {};
    \end{tikzpicture}
    \caption{}
    \label{fig:mutliple_inputs_ex-Jgraph}
\end{subfigure}
\begin{subfigure}[t]{0.28\textwidth}
    \centering
    \includegraphics[height=1.35in]{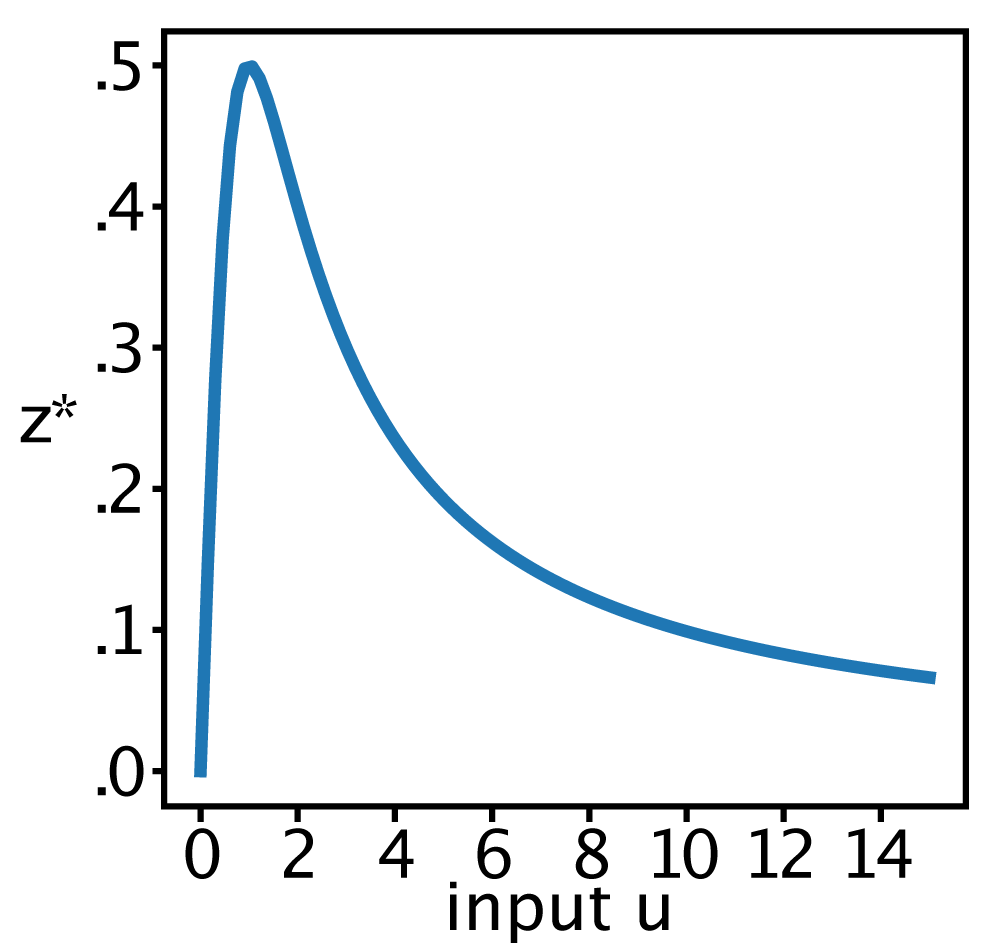} 
    \caption{}
    \label{fig:mutliple_inputs_ex-response}
\end{subfigure}
\begin{subfigure}[t]{0.28\textwidth}
    \centering
    \begin{tikzpicture}
            \node[Jnode] (1) at (0,0) {$x$};  
            \node[Jnode] (2) at (0,-1.5) {$y$}; 
            \node[Jnode] (3) at (1.25,-0.75) {$z$}; 
            \node[Jnode, fill=white] (u) at (-1.25,-0.75) {$x_0$};
            \draw[posEdge] (1)--(3);
            \draw[negEdge] (2)--(3);
            \draw[posEdge] (3)--(2.25,-0.75);
            \draw[posEdge] (u)--(1);
            \draw[posEdge] (u)--(2); 
            \draw[posEdge] (-2.25,-0.75)--(u);
        \node at (0,-2.7) {};
    \end{tikzpicture}
    \caption{}
    \label{fig:mutliple_inputs_ex-Jgraph-extended}
\end{subfigure}
\hspace{1cm}
\caption{
    (a) J-graph of the input-output system \eqref{eq:app-mult_inputs_concrete_original} with multiple inputs.  
    (b) Its steady state response $z^*(u)$ is nonmonotonic. 
    (c) The J-graph of the extended input-output system \eqref{eq:app-mult_inputs_concrete_extended}, where $x_0 = u$. It contains an IFFL, which is necessary for the biphasic response. 
}
\label{fig:mutliple_inputs_ex}
\end{figure}

\ifThisIsArxivVersion
\else 
\vspace{5cm}
\fi
\bibliographystyleSM{IEEEtran}
\bibliographySM{2024_cdc-copy.bib}{}

\end{document}